\documentclass[sigconf, nonacm]{acmart}

\AtBeginDocument{%
  }

\usepackage[most]{tcolorbox}
\tcbuselibrary{listingsutf8}
\usepackage{listings}
\usepackage{numprint}
\usepackage[hang,flushmargin]{footmisc} 
\usepackage{float}
\usepackage{url}
\usepackage{multirow}
\usepackage{enumerate}
\usepackage{enumitem}
\usepackage{appendix}
\usepackage{mathrsfs}
\usepackage{textcomp}
\usepackage{graphicx}
\usepackage{grffile}
\usepackage{mathtools, cuted}
\usepackage{dsfont}
\usepackage{amsthm}
\usepackage{diagbox}
\usepackage{subfig}
\usepackage{xcolor}
\usepackage{algorithmic}
\usepackage[ruled,linesnumbered,noend]{algorithm2e}
\usepackage{adjustbox}
\usepackage{diagbox}
\usepackage{bm}
\usepackage{xspace}
\usepackage{amsmath}
\usepackage{url}
\newtheorem{example}{Example}
\newtheorem{theorem}{Theorem}

\newtheorem{definition}{Definition}
\usepackage{mathtools}
\usepackage{bbm}
\usepackage[utf8]{inputenc} % allow utf-8 input
\usepackage[T1]{fontenc}    % use 8-bit T1 fonts
\usepackage{hyperref}       % hyperlinks
\usepackage{url}            % simple URL typesetting
\usepackage{booktabs}       % professional-quality tables
\usepackage{amsfonts}       % blackboard math symbols
\usepackage{nicefrac}       % compact symbols for 1/2, etc.
\usepackage{microtype}      % microtypography
\usepackage{xcolor}         % colors
\usepackage{amsmath}
\usepackage{makecell}
\usepackage{pifont}
\usepackage{multicol}
\usepackage{balance}
\usepackage[dvipsnames]{xcolor}
\usepackage{threeparttable}
\usepackage[normalem]{ulem}

\PassOptionsToPackage{square,numbers}{natbib}

\floatname{algorithm}{Function}
\floatname{algorithm}{Procedure}

\begin{document}

\title{Shape-Agnostic Table Overlap Discovery: A Maximum Common Subhypergraph Approach}

\author{Ge Lee}
\affiliation{%
  \institution{RMIT University}
  \city{Melbourne}
  \country{Australia}}
\email{ge.lee@student.rmit.edu.au}

\author{Shixun Huang}
\affiliation{%
  \institution{University of Wollongong}
  \city{Wollongong}
  \country{Australia}}
\email{shixunh@uow.edu.au}

\author{Zhifeng Bao}
% \authornote{The corresponding author.}
\affiliation{%
  \institution{The University of Queensland}
  \city{Brisbane}
  \country{Australia}}
\email{zhifeng.bao@uq.edu.au}

\author{Felix Naumann}
\affiliation{%
  \institution{Hasso Plattner Institute}
  \institution{University of Potsdam}
  \city{Potsdam}
  \country{Germany}}
\email{felix.naumann@hpi.de}

\author{Shazia Sadiq}
\affiliation{%
  \institution{The University of Queensland}
  \city{Brisbane}
  \country{Australia}}
\email{s.sadiq@uq.edu.au}

\author{Yanchang Zhao}
\affiliation{%
  \institution{Data61, CSIRO}
  \city{Canberra}
  \country{Australia}}
\email{yanchang.zhao@data61.csiro.au}

%\renewcommand{\shortauthors}{Lee et al.}

% for commenting
\newcommand{\huang}[1]{{\color{Maroon}{\bf{Huang comment:}} #1}}
\newcommand{\bao}[1]{{\color{purple}{\bf{Bao comment:}} #1}}
\newcommand{\lee}[1]{{\color{red}{\bf{Ge comment:}} #1}}

\newcommand{\revised}[1]{{\color{red} #1}}
\newcommand{\todo}[1]{{\color{red} #1}} 

\newcommand{\revII}[1]{{\color{DarkOrchid} #1}}  % reviewer 2
\newcommand{\revIII}[1]{{\color{RoyalBlue} #1}}  % reviewer 3
\newcommand{\revIV}[1]{{\color{BurntOrange} #1}}  % reviewer 4

% for commenting (sigmod revision plan)
\newcommand{\plan}[1]{{\color{Blue} #1}} % proposed revision plan

% pseudocode comment
\newcommand{\commt}[1]{\textit{  // #1}}

% problems
\newcommand{\salto}{\textsf{SALTO}}

% methods
\newcommand{\mcsplit}{\textsf{McSplit}}
\newcommand{\hyper}{\textsf{HyperSplit}}
\newcommand{\sloth}{\textsf{Sloth}}
\newcommand{\gpt}{\textsf{GPT-4.1}}
\newcommand{\crjaccard}{\textsf{CR-Jaccard}}
\newcommand{\jaccard}{\textsf{Jaccard}}

% datasets
\newcommand{\wiki}{WikiTables}
\newcommand{\git}{GitTables}

% colorbox
\newtcolorbox{egbox}{
    colback=gray!10, 
    colframe=white, 
    boxrule=0pt, 
    arc=2mm, 
    left=1.5mm, 
    right=1.5mm, 
    top=1mm, 
    bottom=1mm, 
    enhanced, 
    breakable
}
\newtcolorbox{revbox}{
    colback=gray!10, 
    colframe=white, 
    boxrule=0pt, 
    arc=1.5mm, 
    left=1.5mm, 
    right=1.5mm, 
    top=0.7mm, 
    bottom=0.7mm, 
    enhanced, 
    breakable,
    after skip=0.5em
}
\newtcolorbox{gptpromptbox}{
    colback=gray!5,
    colframe=black!50,
    boxrule=0.5pt,
    arc=2pt,
    outer arc=1pt,
    top=4pt,
    bottom=4pt,
    left=6pt,
    right=6pt,
    enhanced,
    breakable,
    listing only,
    listing options={
        basicstyle=\footnotesize\ttfamily,
        breakindent=0pt,
        breaklines=true,
        columns=fullflexible,
        numbers=none,
    }
}

\lstdefinestyle{promptcompact}{
    basicstyle=\ttfamily\scriptsize, % tighter than \footnotesize
    breaklines=true,
    breakatwhitespace=true,
    columns=fullflexible,
    keepspaces=true,
    showstringspaces=false,
    aboveskip=2pt, belowskip=2pt,
    literate={_}{\_}1 {–}{-}1 {—}{--}2 {“}{``}1 {”}{''}1 {’}{'}1 {…}{...}3
}
\newtcblisting{evbox}{
    breakable,
    enhanced jigsaw,
    colback=gray!3, colframe=gray!55,
    boxrule=0.35pt, arc=2pt,
    left=3pt,right=3pt,top=3pt,bottom=3pt,
    listing only,
    listing options={style=promptcompact},
    before skip=4pt, after skip=4pt,
    sharp corners,
    overlay unbroken and first={\vspace{1pt}},
}
\begin{abstract}
    Understanding how two tables overlap is useful for many data management tasks, but challenging because tables often differ in row and column orders and lack reliable metadata in practice. Prior work defines the \emph{largest rectangular overlap}, which identifies the maximal contiguous region of matching cells under row and column permutations. However, real overlaps are rarely rectangular, where many valid matches may lie outside any single contiguous block. In this paper, we introduce the \emph{Shape-Agnostic Largest Table Overlap} (\salto{}), a novel generalized notion of overlap that captures \emph{arbitrary-shaped, non-contiguous overlaps} between tables. 
    
    To tackle the combinatorial complexity of row and column permutations, we propose to model each table as a hypergraph, casting \salto{} computation into a maximum common subhypergraph problem. We prove their equivalence and show the problem is NP-hard to approximate. To solve it, we propose \hyper{}, a novel branch-and-bound algorithm tailored to table-induced hypergraphs. \hyper{} introduces (\romannumeral 1) hypergraph-aware label classes that jointly encode cell values and their row–column memberships to ensure structurally valid correspondences without explicit permutation enumeration, (\romannumeral 2) incidence-guided refinement and upper-bound pruning that leverage row–column connectivity to eliminate infeasible partial matches early, and (\romannumeral 3) a tolerance-based optimization mechanism with a tunable parameter that relaxes pruning by a bounded margin to accelerate convergence, enabling scalable yet accurate overlap discovery. Experiments on real-world datasets show that \hyper{} discovers overlaps more effectively (larger overlaps in up to 78.8\% of the cases) and more efficiently than state of the art. Three case studies further demonstrate its practical impact across three tasks: cross-source copy detection, data deduplication, and version comparison.
\end{abstract}


\maketitle

% NOTE: For previous version with comments, see archive/2_tex_sigmod/1_introduction_v3.tex, for readability

\section{Introduction}\label{sec:intro}

Tables are a core data representation in databases, web pages, and data lakes. They store structured information about entities, but exhibit substantial structural diversity---columns and rows may appear in different orders or be partially missing, and metadata such as headers are often inconsistent or absent~\cite{adelfio2013schema, farid2016clams, nargesian2019data, zhang2020finding, cafarella2008webtables, wang2023solo, webdatacommons, pimplikar2012answering} (e.g., about 20\% of Web tables lack identifiable headers~\cite{webdatacommons, pimplikar2012answering}). Figures~\ref{subfig:example_a} and~\ref{subfig:example_b} show two Olympic medal tables (content sourced from Wikipedia~\cite{wikipedia_olympic}) containing nearly identical information. Yet, their column orders differ---one starts with athlete names, while the other begins with medal types---and the row orders are inconsistent. Without reordering, most overlapping content remains hidden. Such structural discrepancies and missing headers make it inherently difficult to determine which cells correspond across tables, even when the underlying data are the same.

This raises a fundamental question: \emph{which parts of two tables truly overlap?} In other words, can we recover all matching cells despite arbitrary row and column permutations and in the absence of metadata? \sloth{}~\cite{zecchini2024determining} made the first step toward this goal by defining the \textbf{\emph{largest rectangular-shaped table overlap}} and identifying it using only on cell values. Accurately discovering such overlaps facilitates a range of data management tasks: (i)~\emph{cross-source copy detection}~\cite{li2012truth, li2015scaling, dong2009truth, dong2010global}, identifying and grouping reused or replicated tables across web repositories; (ii)~\emph{data deduplication}~\cite{koch2023duplicate, chu2016data, ilyas2015trends}, which cleans and consolidates overlapping content and allows resolving data inconsistencies~\cite{chu2016data, zecchini2024determining}; and (iii)~\emph{version comparison}~\cite{bleifuss2018exploring}, which accurately localizes overlaps to trace how content evolves across successive table revisions. All these applications rely on the same core operation, that is, exactly determining the shared cells between two tables, which is the focus of this work. 

\smallskip\noindent\textbf{Rectangular Overlap vs. Shape-Agnostic Overlap.}
\sloth{}~\cite{zecchini2024determining} was the first to formalize the problem of finding the \emph{largest rectangular overlap} and remains the state of the art. It defines overlap as a maximal region of matching cells after permuting the rows and columns of both tables. This formulation captures the largest aligned region but imposes a strict rectangular constraint, that is, \emph{all matching cells must lie within one contiguous block spanning consecutive rows and columns}. While effective for perfectly aligned cases, \emph{this restriction fails to capture the irregular, non-contiguous overlap patterns that frequently occur in real data}. 

Ultimately, the definition of overlap itself remains an open question: \emph{is a rectangular notion truly sufficient to represent the full correspondence between two tables?} We argue that a large portion of overlap often lies outside any single rectangle and remains undetected. This argument is further supported by our experimental results and case studies in Sections~\ref{subsec:effectiveness} and~\ref{subsec:case_studies}. Therefore, we remove \sloth{}'s rectangular constraint and \emph{allow overlaps of arbitrary shape}, including disjoint and non-contiguous matches across both dimensions. We call this generalized notion the \textbf{\emph{\uline{S}hape-\uline{A}gnostic \uline{L}argest \uline{T}able \uline{O}verlap} (\salto{})}, which captures those previously hidden correspondences. Example~\ref{ex:largest_overlap} illustrates their difference.

\begin{egbox}
    \begin{example}[]\label{ex:largest_overlap}
        Figures~\ref{subfig:example_a} and~\ref{subfig:example_b} show two tables containing Olympic medal data. At first glance, identifying their overlap is difficult due to differences in row and column ordering, inconsistent data entries and missing column headers. To identify the largest overlap, we reorder the rows and columns of both tables so that the overlapping cells are aligned in position, as shown in Figure~\ref{subfig:example_c}. In this case, the orange-bordered region shows the largest rectangular overlap detected by \sloth{}, covering 3 columns and 6 rows (18 cells). 
        
        \hspace{2mm} However, there are 8 additional overlapping cells outside this rectangular region that are still valid and meaningful matches, such as ``Usain Bolt'', ``JAM'' and others. 
        
        \hspace{2mm} Our shape-agnostic overlap includes both rectangular and non-rectangular, non-contiguous overlaps, yielding a more complete and true intersection comprising 26 cells in total, as shown in the blue-bordered region.
    \end{example}
\end{egbox}

\smallskip\noindent\textbf{Limitations of Other Approaches.}
Astute readers may find several seemingly feasible alternatives, yet these methods \textbf{\emph{do not}} identify exact overlapping cells. One might consider treating tables as multisets of cells and defining their overlap by simple intersection. Another direction is to apply data integration techniques such as schema matching~\cite{hong2002coma, aumueller2005schema, liu2025magneto}, entity resolution~\cite{wu2020zeroer, wang2023sudowoodo}, or unionable table discovery~\cite{nargesian2018table, khatiwada2023santos, fan2023semantics}. Their objectives differ fundamentally and often rely on metadata or contextual cues. In contrast, our goal is to identify \textbf{\emph{exact, cell-level overlap}} purely from cell values. A detailed comparison is provided in Section~\ref{sec:related_work}.

\begin{figure}[t]
    \centering
    \subfloat[Olympic Medal Table A]{
        % \hspace{-0.5em}
        \includegraphics[width=0.2\textwidth]{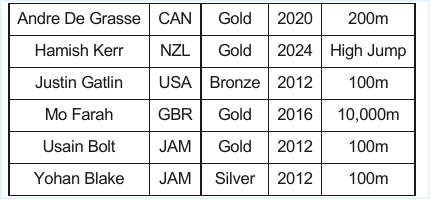}
        \label{subfig:example_a}
    }
    \subfloat[Olympic Medal Table B]{
        \hspace{-1.2em}
        \includegraphics[width=0.27\textwidth]{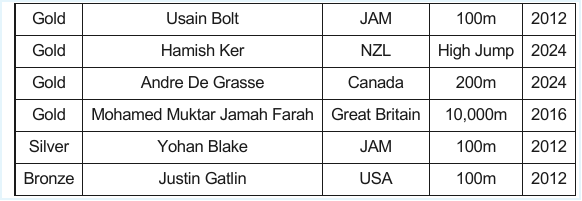}
        \label{subfig:example_b}
    }
    \vfill
    \vspace{-1.2em}
    \subfloat[Largest overlaps after aligning Table A (left) and Table B (right) via row and column permutations]{
        \includegraphics[width=0.47\textwidth]{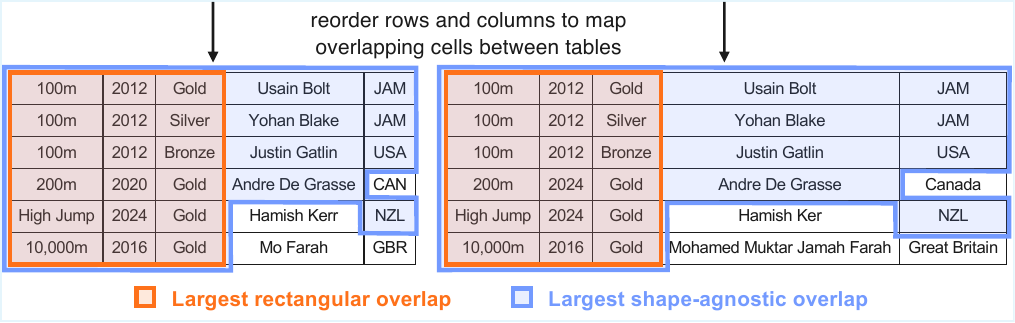}
        \label{subfig:example_c}
    }
    \vspace{0.5em}
    \caption[]{Comparison between \sloth{}'s largest rectangular overlap and our largest shape-agnostic overlap}
    % \vspace{-1em}
\end{figure}

\smallskip\noindent{\textbf{Our Solution.}}
The main computational challenge arises from the factorial growth of the permutation space, which makes exhaustive search infeasible. To overcome this, we formulate each table as a \emph{hypergraph}, casting the computation of \salto{} into a \emph{maximum common subhypergraph} problem between the two table-induced hypergraphs. This formulation avoids explicit permutation enumeration while preserving row and column invariance.

Building on this foundation, we propose \hyper{}, a novel branch-and-bound algorithm designed specifically for table-induced hypergraphs. Standard graphs represent pairwise relationships through binary edges, which cannot capture the dual row–column dependencies inherent in tabular structures. In contrast, our formulation models each table as a hypergraph with higher-order edges that connect multiple cells simultaneously, enforcing joint consistency across both dimensions and introducing structural constraints beyond those in conventional graphs. To address these challenges, \hyper{} introduces three key innovations:
(i)~hypergraph-aware label classes that jointly encode cell values together with their row- and column-hyperedge memberships, ensuring structurally valid correspondences between candidate cells without explicit permutation enumeration;
(ii)~incidence-guided refinement and upper-bound pruning that exploit row–column connectivity to constrain candidate mappings and aggressively prune infeasible search branches; and 
(iii)~a tolerance-based optimization mechanism with a tunable parameter that relaxes the exact pruning condition by a small bounded margin, accelerating convergence while maintaining overlap quality.
Together, these components allow \hyper{} to efficiently identify shape-agnostic overlaps across large and structurally diverse tables.
Our contributions are summarized below:

\begin{itemize}[leftmargin=*]
    \item We propose and formalize, for the first time, the \emph{Shape-Agnostic Largest Table Overlap} (\salto{}) problem, which generalizes rectangular overlap by allowing non-rectangular and non-contiguous matches (Section~\ref{sec:pf}).
    \item We cast the problem of finding \salto{} into that of finding a maximum common subhypergraph by introducing a hypergraph representation for tables (Section~\ref{subsec:prob_reform}) and proving the equivalence between the two formulations (Section~\ref{subsec:proofs}).
    \item We show that the maximum common subhypergraph problem is NP-hard to approximate (Section~\ref{subsec:hardness}).
    \item We propose \hyper{}, an efficient algorithm designed for table-induced hypergraphs that integrates hypergraph-aware labeling, incidence-guided pruning, and tolerance-based optimization to enable scalable and structure-preserving \salto{} detection (Section~\ref{sec:our_solution}).
    \item We conduct comprehensive experiments on real-world datasets, showing that our method identifies overlaps more effectively (larger overlaps in up to 78.8\% of the cases) and more efficiently (Section~\ref{sec:exp}). 
    \item Case studies further demonstrate its practical impact across \emph{cross-source copy detection}, \emph{data deduplication}, and \emph{version comparison} tasks, revealing more complete overlaps without relying on metadata across diverse scenarios  (Section~\ref{subsec:case_studies}).
\end{itemize}
\section{Problem Formulation}\label{sec:pf}

We begin by defining tables, row and column permutations, and the notion of table overlap. We then formalize the shape-agnostic largest table overlap (\salto{}) problem, which seeks the maximum overlap achievable over all possible row and column permutations.

A \emph{table} consists of multiple rows and columns, where each cell contains a value, either a string, a number or null. Given two tables, their \emph{overlap} refers to the set of cells with exactly matching values, including nulls, at the same positions. Since reordering rows and columns does not alter the semantics of a table, we aim to find the \emph{largest possible overlap} achievable through such permutations (see the blue-bordered cells in Figure~\ref{subfig:example_c}). 

\begin{definition}[Table]
    A table $T$ is a two-dimensional structure with $m$ rows and $n$ columns, where each cell $c_{i,j}$ denotes the value at the $i$-th row and the $j$-th column. Formally, $T = \{c_{i, j} \mid 1 \leq i \leq m,\ 1 \leq j \leq n\}$.
\end{definition}

\begin{definition}[Row and Column Permutations]
    Let $I = \{1, \dots, m\}$ denote the row indices of table $T$. A row permutation is a bijection $\sigma: I \to I$ that reorders the row indices by mapping each $i \in I$ to some position $\sigma(i) \in I$. Similarly, a column permutation is a bijection $\tau: J \to J$ for column indices $J = \{1, \dots, n\}$.
\end{definition}

\begin{definition}[Table Overlap]
    Let $T_1$ and $T_2$ be two tables of size $m_1 \times n_1$ and $m_2 \times n_2$, with cell values denoted by $c^{(1)}$ and $c^{(2)}$ respectively. Let $\sigma_1$ and $\sigma_2$ be row permutations for $T_1$ and $T_2$, and $\tau_1$ and $\tau_2$ be the corresponding column permutations. The overlap between $T_1$ and $T_2$ under the permutation tuple $\pi = (\sigma_1, \tau_1, \sigma_2, \tau_2)$ is the set of positions where the permuted cell values are equal. Formally, the permutation-dependent overlap is:
    \[
        O_{\pi} = \left\{ (i, j) \mid c^{(1)}_{\sigma_1(i),\, \tau_1(j)} = c^{(2)}_{\sigma_2(i),\, \tau_2(j)} \right\},
    \]
    where $1 \leq i \leq \min(m_1, m_2)$ and $1 \leq j \leq \min(n_1, n_2)$.
\end{definition}
% \huang{based on def 2, permutation function takes a set as input whereas in def 3, the input is an index, which is inconsistent. one solution: introduce a notation to describe new c after perumtation based def 2.}

\begin{definition}[Shape-Agnostic Largest Table Overlap]\label{def:lolap}
    Let $\mathcal{S}$ be the set of all possible tuples $\pi = (\sigma_1, \tau_1, \sigma_2, \tau_2)$. The shape-agnostic largest table overlap (\salto{}) is the overlap achieved under the permutation $\pi^* \in \mathcal{S}$ that maximizes the overlap size, i.e., $O_{\pi^*}, \text{where } \pi^* = \arg\max_{\pi \in \mathcal{S}} |O_\pi|.$
\end{definition}

Definition~\ref{def:lolap} maximizes the number of exactly matching cells under a single globally consistent row/column permutation tuple $\pi^*$. Thus, $O_{\pi^*}$ is not required to be rectangular and may be unconnected, but it is not arbitrary because every matched cell must respect the same alignment. \salto{} optimizes overlap size only and does not enforce geometric compactness or a rectangular prefix. If a use case requires a rectangular prefix, one could define a constrained variant that maximizes $|O_\pi|$ subject to containing an $m'\!\times\! n'$ rectangle. The rectangle-constrained \sloth{} searches the same permutation space but restricts the overlap to one rectangle, so its optimum size is no larger than $|O_{\pi^*}|$. However, the specific rectangle returned by \sloth{} need not be contained in $O_{\pi^*}$ because the two objectives may select different optimal permutations.

\section{Casting \salto{} to Maximum Common Subhypergraph}\label{sec:cast}

Consider two tables with potentially different numbers of rows and columns. Identifying their \salto{} requires exploring all rows and columns permutations in both tables, resulting in a worst-case search space of $\mathcal{O}(m_1! \cdot n_1! \cdot m_2! \cdot n_2!)$. If one table has fewer rows and columns than the other, we can fix its order and permute only the larger table, reducing the space to $\mathcal{O}(m! \cdot n!)$. Nevertheless, even this reduced space remains prohibitively large in practice. This combinatorial explosion motivates a structural reformulation that eliminates the need for exhaustive permutation search.

To address this challenge, we propose a novel representation of tables as hypergraphs, where each common table cell value, i.e., a cell value that appears in both tables, is modeled as a node, and rows and columns are encoded as two separate sets of hyperedges. This transformation offers three key benefits: (\romannumeral 1)~it fully preserves the structural properties of the tables; (\romannumeral 2)~it eliminates the need to explicitly search over the extremely large combinatorial space of row and column permutations; (\romannumeral 3)~it naturally reduces the computation to only those cells that appear in both tables, which are the only possible candidates for forming the overlap. Under this representation, finding the \salto{} is equivalent to finding a maximum common subhypergraph between the two hypergraphs. 

In what follows, we cast the \salto{} problem into the maximum common subhypergraph problem. Section~\ref{subsec:prob_reform} constructs hypergraphs from tables and outlines how table structure is captured, Section~\ref{subsec:mcs} then formulates the maximum common subhypergraph problem on the resulting hypergraphs, and Section~\ref{subsec:proofs} proves the equivalence between the two formulations, completing the casting.

\begin{figure}[t]
    \centering
    \includegraphics[width=0.45\textwidth]{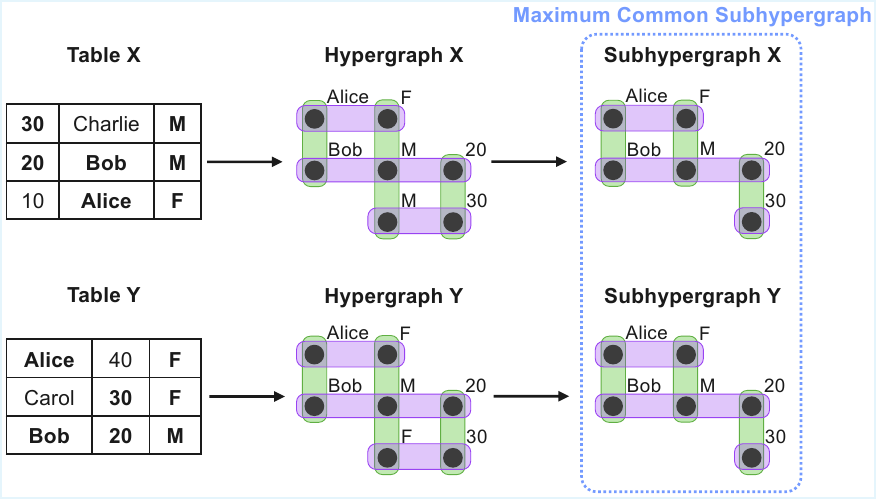}
    \caption{Hypergraph representation of Tables X and Y and their maximum common subhypergraph. Cell values shared by both tables are shown in bold. Hyperedges are visualized as convex hulls (purple for rows, green for columns) enclosing their associated nodes.}\label{fig:hypergraph}
    % \vspace{-1em}
\end{figure}

\subsection{Table-to-Hypergraph Formulation}\label{subsec:prob_reform}

To formulate the \salto{} problem as a maximum common subhypergraph problem, we first define hypergraphs and then describe how to construct them to represent tables in this context.

\subsubsection{Hypergraph}

Edges in a hypergraph can connect multiple nodes at once, allowing for more complex relationships to be represented compared to a standard graph. Hypergraphs can be directed or undirected. Here, we make use only of \emph{undirected} hypergraphs.

\begin{definition}[Hypergraph]
    A hypergraph is defined as $H = (V, E)$, where $V = \{v_1, \dots, v_{|V|}\}$ is a finite set of nodes, and $E = \{e_1, \dots, e_{|E|}\}$ is a finite set of hyperedges, each of which is a non-empty subset of $V$. Formally, $E \subseteq \mathcal{P}(V) \setminus \{\emptyset\}$, where $\mathcal{P}(V)$ denotes the power set of $V$.
\end{definition}

\subsubsection{Table-to-Hypergraph Construction}\label{subsubsec:t2h_construction}

For each table in the given pair, we construct a corresponding hypergraph. The construction procedure is outlined below:

\smallskip\noindent\textbf{Step 1: Node Construction (Cross-Table Dependency)}. We begin by identifying all \emph{common cell values}, i.e., values that appear in \emph{both} tables. For each occurrence of a common cell value in a table, we create a corresponding node $v$ in that table's hypergraph. Thus, if a common cell value appears multiple times in one table, we create one node for each occurrence, all labeled with the same value. This preserves multiplicity within each table and naturally supports many-to-many cell mappings between repeated values across tables. Specifically, each node representing a cell value in one hypergraph can potentially match any node with the same label in the other hypergraph. By restricting node creation to only common cell values, we introduce a cross-table dependency from the outset, directly coupling the two tables and reducing the search space to only those cells that can contribute to a potential overlap.

\begin{egbox}
    \begin{example}[]
        Consider the example in Figure~\ref{fig:hypergraph}. We first identify all common cell values, which are highlighted in bold in Tables X and Y. For each occurrence of a common cell value in a table, we create a corresponding node in that table's hypergraph. For example, the common cell value ``M'' appears twice in Table X, so two distinct nodes labeled ``M'' are created in Hypergraph X. Similarly, ``F'' appears twice in Table Y, resulting in two nodes labeled ``F'' in Hypergraph Y. All other common cell values appear only once in each table, so a single node is created for each in the respective hypergraphs.
    \end{example}
\end{egbox}

\smallskip\noindent\textbf{Step 2: Hyperedge Construction (Intra-Table Dependency)}. For each row in a table, we create a row-hyperedge, containing exactly those nodes corresponding to the common cells in that row. Similarly, for each column, we create a column-hyperedge with the corresponding nodes. This construction captures intra-table dependencies by explicitly modeling the structural relationships among cells within the same row or column. Consequently, each node always belongs to exactly one row-hyperedge and one column-hyperedge, preserving the structural organization of the table.

\begin{egbox}
    \begin{example}[]
        In Figure~\ref{fig:hypergraph}, row-hyperedges (highlighted in purple) connect all nodes from the same row, while column-hyperedges (highlighted in green) connect all nodes from the same column. For example, in the first row of Table X, the common cell values ``30'' and ``M'' correspond to nodes connected by a purple row-hyperedge in Hypergraph X, representing their relationship within that row. Similarly, in the first column of Table X, the common cell values ``30'' and ``20'' are connected by a green column-hyperedge, capturing the column-level structure.
    \end{example}
\end{egbox}

\noindent\textbf{Complexity.} Let $T_1$ and $T_2$ be two input tables of sizes $m_1\cdot n_1$ and $m_2\cdot n_2$, respectively. Constructing the nodes requires scanning both tables once to identify all cell occurrences whose values appear in both tables, taking $\mathcal{O}(m_1n_1 + m_2n_2)$ time. We then group these nodes into row- and column-hyperedges by iterating through rows and columns, which takes $\mathcal{O}(m_1 + n_1 + m_2 + n_2)$ time and creates $\mathcal{O}(m_1 + n_1)$ and $\mathcal{O}(m_2 + n_2)$ hyperedges for the two hypergraphs. Overall, the transformation runs in $\mathcal{O}(m_1n_1 + m_2n_2)$ time, i.e., linear in the total number of table cells. We materialize each hypergraph as a sparse incidence structure over only common-value cell occurrences. Each node stores its row and column membership. Thus, the memory is linear in the number of common-value cell occurrences $|V|$, which is upper-bounded by the total number of cells $m_1n_1+m_2n_2$ (e.g., when every cell in both tables contains the same value). Compared to the factorial search space $\mathcal{O}(m_1! \cdot n_1! \cdot m_2! \cdot n_2!)$ of row and column permutations, this construction step is significantly~more~efficient.

\subsection{A New Problem: Maximum Common Subhypergraph}\label{subsec:mcs}

Having constructed hypergraphs from tables, we now formalize the new problem defined on these structures. Specifically, we introduce the definitions of subhypergraphs, common subhypergraphs, and the maximum common subhypergraph problem on the table-induced hypergraphs.

\subsubsection{Subhypergraph}

An induced \emph{subhypergraph} is obtained by selecting a subset of nodes from the original hypergraph and restricting each hyperedge to its intersection with this subset. In other words, we retain all hyperedges that contain at least one selected node from the subset and trim each to include only the selected nodes it contains. For simplicity, we refer to an induced subhypergraph as a subhypergraph for the rest of this paper.

\begin{definition}[Subhypergraph]
    Given a hypergraph $H = (V, E)$, the subhypergraph induced by a subset of nodes $V' \subseteq V$ is defined as $H[V'] = (V', E')$, where $E' = \{ e \cap V' \mid e \in E, e \cap V' \neq \emptyset \}$.
\end{definition}

\begin{egbox}
    \begin{example}[]
        Consider the example of Hypergraph X and Subhypergraph X in Figure~\ref{fig:hypergraph}. To form Subhypergraph X, we select all nodes from Hypergraph X except the bottom node labeled ``M''. Hyperedges that previously connected the excluded node are now trimmed to connect only the remaining nodes, i.e., one row-hyperedge is reduced to \{``30''\} and one column-hyperedge becomes \{``F'', ``M''\}. The remaining hyperedges remain unchanged. 
    \end{example}
\end{egbox}

\subsubsection{Common Subhypergraph}\label{subsubsec:common_subgraph}

A \emph{common subhypergraph} is a hypergraph that is isomorphic to subhypergraphs of each of the two hypergraphs. Specifically, there exist bijections between their nodes and hyperedges that preserve incidence. We refer to these corresponding elements as matched nodes and matched hyperedges.

\begin{definition}[Common Subhypergraph]\label{def:cs}
    Let $H_1 = (V_1, E_1)$ and $H_2 = (V_2, E_2)$ be two hypergraphs. A hypergraph $C = (V_C, E_C)$ is a common subhypergraph of $H_1$ and $H_2$ if there exist subsets $V_1' \subseteq V_1$ and $V_2' \subseteq V_2$ such that $C$ is isomorphic to both subhypergraphs $H_1[V_1'] = (V_1', E_1')$ and $H_2[V_2'] = (V_2', E_2')$. That is, there exist bijections $\phi_V : V_C \to V_i'$ and $\phi_E : E_C \to E_i'$ (for $i = 1,2$) such that for all $e \in E_C$ and $v \in V_C$, we have $v \in e$ if and only if $\phi_V(v) \in \phi_E(e)$.
\end{definition}

\noindent\textbf{Key Constraints.} In our setting, a valid common subhypergraph must satisfy two additional constraints: (1)~matched nodes between the two hypergraphs must share the same label, reflecting overlapping cell values in the tables; (2)~each matched node must belong to exactly one row-hyperedge and one column-hyperedge, following the construction in Step~2 (Section~\ref{subsubsec:t2h_construction}), thereby preserving the structural properties of the tables.

\begin{egbox}
    \begin{example}[]
        Consider the example in Figure~\ref{fig:hypergraph}. Both Subhypergraphs X and Y form a common subhypergraph (they also constitute a maximum one, as we explain next). The nodes correspond one-to-one between both subhypergraphs with matching labels, and their row- and column-hyperedges are preserved consistently, e.g., the hyperedge with nodes \{``Alice'', ``F''\} in Subhypergraph X matches its counterpart in Subhypergraph Y. Additionally, both subhypergraphs satisfy the two key constraints required for a valid common subhypergraph.
    \end{example}
\end{egbox}

\subsubsection{Maximum Common Subhypergraph}

The \emph{maximum common subhypergraph} is a common subhypergraph with the maximum number of nodes, along with their corresponding hyperedges.

\begin{definition}[Maximum Common Subhypergraph]
    Let $H_1$ and $H_2$ be two hypergraphs. A hypergraph $C^* = (V^*_C, E^*_C)$ is a \emph{maximum common subhypergraph} of $H_1$ and $H_2$ if it is a common subhypergraph of both, and for any common subhypergraph $C = (V_C, E_C)$, it holds that $|V_C| \leq |V^*_C|$.
\end{definition}

\subsection{Establishing Problem Equivalence}\label{subsec:proofs}

We now establish the formal equivalence between the two formulations. Specifically, we show that solving the maximum common subhypergraph problem on hypergraphs constructed from tables yields valid and complete solutions to the \salto{} problem. To prove this equivalence, we first demonstrate that: (i)~the hypergraph construction is invariant under row and column permutations; (ii)~the solution to the maximum common subhypergraph problem corresponds exactly to the overlap in the \salto{} problem.

\smallskip\noindent\textbf{Row and Column Permutation Invariance in Hypergraph Construction.}
Permuting the rows or columns of a table results in hypergraphs that are isomorphic. This invariance ensures that the hypergraph structure remains unchanged under row or column permutations, preserving structural properties including node-hypergraph incidence and hyperedge cardinalities.

\begin{theorem}[Permutation Invariance]\label{theorem:perm_inv}
    Let $T$ be a table and $H(T)$ be its associated hypergraph. For any table $T'$ obtained by permuting the rows and/or columns of $T$, the hypergraphs $H(T)$ and $H(T')$ are isomorphic. Consequently, structural properties of the hypergraph representation are preserved under such permutations.
\end{theorem}

\begin{proof}
    \emph{Proof by construction.} Let $T$ be a table with rows indexed by $I$ and columns by $J$. Let $T'$ be the table obtained from $T$ by applying a row permutation $\sigma$ to the rows and a column permutation $\tau$ to the columns, that is, $T'(\sigma(i),\tau(j)) = T(i,j)$ for all $(i,j)\in I\times J$.
    
    \smallskip\noindent\uline{Node Set.} We construct the hypergraph $H(T)$ by creating one node $v_{i,j}$ for each cell $(i,j)$ of $T$.
    Since row and column permutations do not change the multiset of cell values, the node sets of $H(T)$ and $H(T')$ coincide, that is, $V = V' = \{v_{i,j} \mid (i,j) \in I \times J\}$.
    
    \smallskip\noindent\uline{Hyperedge Set.} In $H(T)$, we define $i$-th row-hyperedge $e_i = \{v_{i,j}\mid j\in J\} \text{ for each } i\in I$, and $j$-th column-hyperedge $f_j = \{v_{i,j}\mid i\in I\} \text{ for each } j\in J$. Similarly in $H(T')$, we define $e'_{\sigma(i)} = \{v_{\sigma(i),\tau(j)} \mid j \in J\}$ and $f'_{\tau(j)} = \{v_{\sigma(i),\tau(j)} \mid i \in I\}$.
    
    \smallskip\noindent\uline{Isomorphism.} Let $\phi_V:V\to V'$ be the identity map on the node set, since $\phi_V(v) = v$. We define the hyperedge bijection $\phi_E$ by $\phi_E(e_i) = e'_{\sigma(i)}$ and $\phi_E(f_j) = f'_{\tau(j)}$. For every node $v$ and hyperedge $e \in \{e_i,f_j\}$, we have $v \in e \iff \phi_V(v) \in \phi_E(e)$, hence $(\phi_V,\phi_E)$ is an incidence-preserving pair of bijections, i.e., an isomorphism. Therefore, $H(T)$ and $H(T')$ are isomorphic hypergraphs.
    
    Since hypergraph isomorphism preserves all structural properties including node-hyperedge incidence and hyperedge sizes, the proof is complete.
\end{proof}

\smallskip\noindent\textbf{Maximum Common Subhypergraph as \salto{}.}
The maximum common subhypergraph between two table-induced hypergraphs corresponds exactly to the largest table overlap in \salto{}, establishing the formal equivalence between the two formulations.

\begin{theorem}[Overlap Equivalence]
    Let $T_1$ and $T_2$ be two tables and $H(T_1)$ and $H(T_2)$ be their corresponding hypergraphs. The maximum common subhypergraph of $H(T_1)$ and $H(T_2)$ is equivalent to the \salto{} between $T_1$ and $T_2$.
\end{theorem}

\begin{proof}
    \emph{Proof by construction}. 

    \smallskip\noindent\uline{Direction 1: From Table Overlap to Common Subhypergraph.} Recall from Definition~\ref{def:lolap} that $O_{\pi^*}$ is the \salto{} between $T_1$ and $T_2$ achieved under optimal row and column permutations. Each cell $c$ in a table corresponds to a node $v$ in the associated hypergraph. Therefore, the set of overlapping cells defines a set of nodes $V_C \subseteq V_1 \cap V_2$, where each node $v \in V_C$ corresponds to a matching cell in the overlap. The hypergraph construction includes one hyperedge per row and column, connecting all cells in the respective row or column. Hence, the rows and columns in the overlapping cells induce structurally identical row- and column-hyperedges in both $H(T'_1)$ and $H(T'_2)$. Let $C$ be the subhypergraph induced by the nodes in $V_C$, along with their corresponding row- and column-hyperedges. Since the overlap $O_{\pi^*}$ is maximal by construction, $C$ is the maximum common subhypergraph of $H(T'_1)$ and $H(T'_2)$.
    
    \smallskip\noindent\uline{Direction 2: From Common Subhypergraph to Table Overlap.} Conversely, let $C^*$ be the maximum common subhypergraph of $H(T_1)$ and $H(T_2)$ with node set $V_C$. By the construction of the hypergraphs, each node $v \in V_C$ corresponds to matching cells in $T_1$ and $T_2$. The hyperedges of $C^*$ correspond to the rows and columns containing these overlapping cells. By Theorem~\ref{theorem:perm_inv}, the hypergraph structure is invariant under row and column permutations. Therefore, we can permute the rows and columns of $T_1$ and $T_2$ such that the corresponding cells associated with $V_C$ are aligned in the same positions in both tables. This alignment yields a table overlap of size $|V_C|$. Suppose, for contradiction, that there exists an overlap larger than $|V_C|$. Then, by the construction in Direction 1, this larger overlap would induce a common subhypergraph with more nodes than $C^*$, contradicting its maximality. Hence, the overlap induced by $C^*$ is the largest possible.
\end{proof}
\section{Hardness Analysis}\label{subsec:hardness}

We show that the maximum common subhypergraph problem is NP-hard and NP-hard to approximate. By the established equivalence between the maximum common subhypergraph and \salto{} problems, it follows that \salto{} inherits the same computational hardness.

\begin{theorem}\label{theorem:nphard}
    Finding the maximum common subhypergraph problem is NP-hard.
\end{theorem}
\vspace{-1em}

\begin{proof}
    We prove the theorem via a reduction from the maximum clique problem, which is NP‐hard~\cite{karp2009reducibility}.

    \begin{definition}[Maximum Clique Problem~\cite{karp2009reducibility}]\label{def:maxclique}
        Let $G = (V, E_G)$ be an undirected graph, where $V$ is a set of nodes and $E_G$ is a set of edges. A clique is a subset $K \subseteq V$ such that every pair of distinct nodes in $K$ is connected by an edge, i.e., $\{u, v\} \in E_G$ for all $u, v \in K$ with $u \neq v$. The goal is to find a clique $K \subseteq V$ of maximum size, i.e., one that maximizes $|K|$ in $G$.
    \end{definition}
    
    Let $G = (V, E_G)$ be an instance of the maximum clique problem. We construct two hypergraphs $H_1$ and $H_2$ as follows:
    \begin{enumerate}[leftmargin=*]
        \item $H_1 = (V, E_1)$, where $E_1 = \{\{u, v\} \mid (u, v) \in E_G\}$. In this case, $H_1$ is simply the original graph $G$ viewed as a hypergraph, where each hyperedge consists of exactly two nodes, corresponding to the edges of the original graph. 
        \item $H_2 = (V, E_2)$, where $E_2 = \{\{u, v\} \mid u, v \in V, u \neq v\}$. Here, $H_2$ represents the complete graph on the node set $V$, which is also viewed as a hypergraph with every pair of distinct nodes forming a hyperedge.
        % Recall that a complete graph is a graph in which every pair of distinct nodes is connected by an edge.
    \end{enumerate}
    
    Consider any subset $V_C \subseteq V$. By Definition~\ref{def:cs}, a common subhypergraph $C = (V_C, E_C)$ of $H_1$ and $H_2$ is an induced subhypergraph on $V_C$, where $E_C$ contains all hyperedges in both $H_1$ and $H_2$ when restricted to $V_C$. Since $H_2$ is the complete graph, any pair of distinct nodes in $V_C$ is present in $E_2$. Thus, $E_C$ consists of those hyperedges in $E_1$ that are induced by $V_C$. For $E_C$ to contain all possible pairs in $V_C$, every such pair must also exist in $E_1$, i.e., $V_C$ must form a clique in $G$. Thus, finding a maximum common subhypergraph of $H_1$ and $H_2$ is equivalent to finding a maximum clique in $G$. Since the construction of $H_1$ and $H_2$ from $G$ is achievable in polynomial time, if there is a polynomial-time algorithm that finds the optimal solution $C^*$ to our problem, then we can solve the maximum clique problem optimally. This is not possible unless $\text{P} = \text{NP}$. Therefore, the maximum common subhypergraph problem is NP‐hard.
    % \huang{you may want to say, if there is a polynomial time optimal algorithm to our problem, we can solve maximum clique optimally. this is not true, unless ....}
\end{proof}

\begin{theorem}
    The maximum common subhypergraph problem is NP-hard to approximate within a factor of $\mathcal{O}(|\mathcal{V}|^{1 - \varepsilon})$ for any $\varepsilon > 0$ in polynomial time.
\end{theorem}

\begin{proof}
    The reduction in the proof for Theorem~\ref{theorem:nphard} is approximation-preserving. Specifically, any common subhypergraph $C$ of size $|K|$ corresponds directly to a clique of size $|K|$ in $G$. Therefore, an algorithm that approximates the maximum common subhypergraph problem within a factor better than $|V|^{1 - \varepsilon}$ would yield an approximation algorithm for the maximum clique problem with the same guarantee. However, it is known that the maximum clique problem cannot be approximated within a factor of $|V|^{1 - \varepsilon}$ for any $\varepsilon > 0$ unless $\text{P} = \text{NP}$~\cite{hastad1996clique}. This contradiction implies that no such approximation algorithm exists for the maximum common subhypergraph problem.
\end{proof}

\section{\hyper{}: Solving Maximum Common Subhypergraph}\label{sec:our_solution}

This section presents our proposed algorithm, \hyper{}, for solving the maximum common subhypergraph problem induced by our table-to-hypergraph reformulation. We first introduce the background and limitations of the existing graph-based method \mcsplit{}~\cite{mccreesh2017partitioning} (Section~\ref{subsec:back_mcp}), which inspires our method. We then present the core steps of \hyper{}, each of which addresses a specific challenge in adapting graph-based matching to the hypergraph setting (Sections~\ref{subsec:hyper_overview}--\ref{subsec:adaptations}). Finally, we introduce tolerance-based pruning, an optimization that accelerates the search process with approximate solutions (Section~\ref{subsec:tolerance}).

\subsection{Background: \mcsplit{} and Its Limitations}\label{subsec:back_mcp}

\mcsplit{}~\cite{mccreesh2017partitioning} is a well-established branch-and-bound algorithm for finding the largest common subgraph between two graphs. It operates by maintaining a set of valid node pairings, grouped into label classes, and recursively explores these pairings using backtracking. At each step, \mcsplit{}: (1)~selects a label class and a node pair to match; (2)~updates the label classes based on adjacency constraints; (3)~computes an upper bound to prune unpromising branches.

However, \mcsplit{} assumes standard graphs with binary edges and adjacency-based compatibility. Our problem involves table-induced hypergraphs that differ fundamentally from standard graphs in three key ways: (i)~\emph{Hyperedges connect multiple nodes simultaneously.} Unlike binary edges, each hyperedge corresponds to an entire table row or column, connecting many nodes; (i)~\emph{Each node belongs to two distinct row- and column-hyperedges.} Valid mappings must preserve this dual membership, ensuring structural consistency across both dimensions; (iii)~\emph{Node mappings must preserve both cell values and hyperedge incidence.} Mapped nodes must share the same cell values and come from matched row-to-row and column-to-column hyperedges. These differences reveal that \mcsplit{}’s graph-based assumptions do not capture the complexities of hypergraph structures, motivating the development of \hyper{} to effectively handle these non-trivial constraints.

\subsection{Overview of \hyper{}}\label{subsec:hyper_overview}

% \huang{you need to describe what is mcsplit first such that people can know why we can get inspiration from it.}
We propose \hyper{}, a recursive branch-and-bound algorithm to optimally solve the maximum common subhypergraph problem between the two table-induced hypergraphs $H_1$ and $H_2$ from our construction. While inspired by the partition-based \mcsplit{} framework~\cite{mccreesh2017partitioning}, \hyper{} introduces non-trivial adaptations to handle the dual row/column hyperedges incidence and the multiplicities of common cell values that are induced by our construction. These adaptations are essential, as the off-the-shelf \mcsplit{} cannot enforce the structural constraints posed by our hypergraph formulation and thus fails to solve the problem. Accordingly, \hyper{} is specialized to these table-induced hypergraphs and does not aim to be a general-purpose solver for arbitrary hypergraphs. We first present the algorithmic pipeline, followed by highlighting the non-trivial adaptations specifically designed to address these constraints.

The pseudocode for \hyper{} is shown in Algorithm~\ref{algo:hypermcsplit}. The algorithm incrementally builds a node-to-node mapping $M \subseteq V_1 \times V_2$ between two hypergraphs $H_1 = (V_1, E_1)$ and $H_2 = (V_2, E_2)$ using depth-first search with branch-and-bound pruning. It starts from an empty mapping set $M$ (line~\ref{algo_hmcsp:3}) and explores candidates one pair at a time. To manage candidates efficiently, unmapped nodes are grouped into \emph{hypergraph-aware label classes} (line~\ref{algo_hmcsp:4}) based on their cell values and hyperedge incidences. At each search state, it \emph{selects a candidate node pair} $(v_1$, $v_2)$, checks whether it satisfies hypergraph-specific compatibility conditions, and extends the mapping accordingly (line~\ref{algo_hmcsp:12}-\ref{algo_hmcsp:14}). If the current mapping can potentially lead to a larger $M$ than the best solution found so far $|M^*|$, the search continues recursively. Otherwise, it backtracks using an upper bound estimate (line~\ref{algo_hmcsp:9}-\ref{algo_hmcsp:10}). The best solution is updated throughout the process and returned upon termination (line~\ref{algo_hmcsp:8}). The algorithm terminates once all branches have been explored. Although we focus on the maximum solution, the formulation represents every feasible overlap as a common subhypergraph, and the search traverses many such candidates. The same framework can be adapted for enumeration variants, such as returning the top-$k$ overlaps by size or those above a size threshold, which correspond to possibly different valid row and column rearrangements and overlap shapes.
% \huang{it is unclear to know the logic below. what is the relationship between the following graphs and the overview above? also, there is no roadmap for this whole section. people do not know the writing logic.}

\begin{algorithm}[t]
    \caption{\hyper{}}\label{algo:hypermcsplit}
    \SetKwInOut{Input}{Input}
    \SetKwInOut{Output}{Output}
    \SetKwFunction{FMain}{Main}
    \SetKwFunction{FSearch}{Search}
    \SetKwProg{Fn}{}{:}{}
    \small
    \Input{Two hypergraphs $H_1$ and $H_2$.}
    \Output{$M^*$.}
    $M^* \leftarrow \emptyset$; \commt{best solution}\\
    \Fn{\FMain{$G_1, G_2$}} {
        $M \leftarrow \emptyset$; \commt{current solution}~\label{algo_hmcsp:3}\\
        $\mathcal{U} \leftarrow \text{group } V_1, V_2 \text{ into label classes based on } \lambda_{H_1},\, \lambda_{H_2}$~\label{algo_hmcsp:4}\;
        \FSearch{$M, \mathcal{U}$}\;
        \Return{$M^*$}\;
    }
    \Fn{\FSearch{$M, \mathcal{U}$}} {
        $\textbf{if } |M| > |M^*| \textbf{ then } M^* \leftarrow M$~\label{algo_hmcsp:8}\;
        $UB \leftarrow \text{compute upper bound (Equation~\ref{eq:ub})}$~\label{algo_hmcsp:9}\;
        $\textbf{if } UB \leq |M^*| \textbf{ then return}$~\label{algo_hmcsp:10}\;
        $(U_1, U_2) \leftarrow \text{select a label class from } \mathcal{U} \text{ (Section~\ref{itm:h1})}$~\label{algo_hmcsp:11}\;
        $v_1 \leftarrow \text{select a node from } U_1 \text{ (Section~\ref{itm:h2})}$~\label{algo_hmcsp:12}\;
        \For {$v_2 \in U_2$} {
            $M_{\text{new}} \leftarrow M \cup \{(v_1, v_2)\}$~\label{algo_hmcsp:14}\;
            $\mathcal{U}_{\text{new}} \leftarrow \text{update label classes } \mathcal{U} \text{ based on } M_\text{new}$~\label{algo_hmcsp:15}\;
            \FSearch{$M_{\text{new}}, \mathcal{U}_{\text{new}}$}\;
        }
        $U'_1 \leftarrow U_1 \setminus \{v_1\}$\;
        $\mathcal{U} \leftarrow \mathcal{U} \setminus \{(U_1, U_2)\}$\;
        $\textbf{if } U' \neq \emptyset \textbf{ then } \mathcal{U} \leftarrow \mathcal{U} \cup \{(U'_1, U_2)\}$\;
        \FSearch{$M, \mathcal{U}$}\;
    }
\end{algorithm}

\subsection{Candidate Node Pair Selection}\label{subsec:cand_sel}

At any point during the search, the current set of mapped node pairs is denoted as $M = \{(v^1_1,\, v^1_2), \dots, (v^{|M|}_1,\, v^{|M|}_2)\}$. A pair of unmapped nodes $(u_1, u_2) \in (V_1 \setminus \{v^1_1, \dots, v^{|M|}_1\}) \times (V_2 \setminus \{v^1_2, \dots, v^{|M|}_2\})$ is considered a candidate if it satisfies the following three conditions:

\smallskip\noindent\textbf{(i)~Cell Value Equality.}
The two unmapped nodes $u_1$ and $u_2$ must have the same associated cell value $\ell$ from their tables to indicate a potential pair of overlapping cells, i.e., $\ell(u_1) = \ell(u_2)$.

\smallskip\noindent\textbf{(ii)~Hyperedge Compatibility.}
The mapping must preserve the structural consistency of row-hyperedges and column-hyperedges with already mapped nodes. Let $E^\text{row}_1$, $E^\text{col}_1$ be the row and column hyperedges in $H_1$, and $E^\text{row}_2$, $E^\text{col}_2$ be the row and column hyperedges in $H_2$. Then for each previously mapped pair $(v_1, v_2) \in M$, the pair $(u_1, u_2)$ must maintain the pairwise hyperedge structure. That is, whenever $u_1$ and $v_1$ both belong to a hyperedge in $H_1$, then $u_2$ and $v_2$ must belong to the corresponding hyperedge in $H_2$, and vice versa. More formally:
\[
    \begin{aligned}
    \quad (\exists e_1 \in E^\text{row}_1\ |\ \{u_1, v_1\} \subseteq e_1) &\iff (\exists e_2 \in E^\text{row}_2\ |\ \{u_2, v_2\} \subseteq e_2), \\
    (\exists e_1 \in E^\text{col}_1\ |\ \{u_1, v_1\} \subseteq e_1) &\iff (\exists e_2 \in E^\text{col}_2\ |\ \{u_2, v_2\} \subseteq e_2).
    \end{aligned}
\]
This ensures that $u_1$ and $u_2$ always belong to structurally equivalent row and column hyperedges with respect to the nodes that are already mapped. It enables the algorithm to incrementally preserve the hyperedge structure, adding one node pair at a time, without the need to explicitly enumerate entire hyperedges during the search.

To implement this efficiently, we maintain two auxiliary hash maps, one for rows and one for columns, that record the correspondence between hyperedges in the two graphs. These maps ensure that a pair of candidate nodes can only be added if their row- and column-hyperedges are either already matched by the current mapping, or both still unmatched and thus allowed to be matched. These maps are updated in $\mathcal{O}(1)$ time whenever a new node pair is added to $M$ and ensure that the mapping can always be extended to a consistent row and column isomorphism.

\smallskip\noindent\textbf{(iii)~Unique Hyperedge Membership.}
By construction, each node $u_1 \in V_1$ and $u_2 \in V_2$ belongs to exactly one row-hyperedge and one column-hyperedge only. This constraint enforces that each node corresponds to a unique cell in a table (exactly one row and one column), reflecting the underlying table structure encoded in the hypergraphs.

\subsection{Hypergraph-Aware Label Class}

To manage the candidates efficiently during the search, we group unmapped nodes into \emph{label classes} based on their cell values and hyperedges. For each unmapped node $u_1 \in V_1 \setminus \{v^1_1, \dots, v^{|M|}_1\}$, we define its label as $\lambda_{H_1}(u_1) = (\ell(u_1), \alpha_1(u_1), \beta_1(u_1))$, where:
\begin{itemize}[leftmargin=*]
    \item $\ell(u_1)$ is the cell value of node $u_1$,
    \item $\alpha_1(u_1) = [\mathbbm{1}_{(u_1,\, v^i_1) \in e} \text{ for } e \in E^\text{row}_1]^{|M|}_{i=1}$ encodes row-hyperedge incidence with each mapped node $v^i_1$,
    \item  $\beta(u_1) = [\mathbbm{1}_{(u_1,\, v^i_1) \in e} \text{ for } e \in E^\text{col}_1]^{|M|}_{i=1}$ encodes column-hyperedge incidence.
\end{itemize}
Similarly, for each $u_2 \in V_2 \setminus \{v^1_2, \dots, v^{|M|}_2\}$, we define $\lambda_{H_2}(u_2) = (\ell(u_2), \alpha_1(u_2), \beta_2(u_2))$ with $\alpha_1(u_2)$ and $\beta_2(u_2)$ defined analogously over $E^\text{row}_2$ and $E^\text{col}_2$ respectively.
A hyperedge-aware label class is then defined as a pair $(U_1, U_2)$, where all nodes in $U_1$ and $U_2$ share the same label $\lambda$, i.e., 
\[
    \begin{aligned}
        U_1 &= \{u_1 \in V_1 \setminus \{v^1_1, \dots, v^{|M|}_1\} \mid \lambda_{H_1}(u_1) = \lambda \}, \\
        U_2 &= \{u_2 \in V_2 \setminus \{v^1_2, \dots, v^{|M|}_2\} \mid \lambda_{H_2}(u_2) = \lambda \}.
    \end{aligned}
\]
Only node pairs $(u_1, u_2) \in U_1 \times U_2$ of the same class are considered to extend the mapping set $M$. This guarantees that the associated cell values and hyperedge incidence remain consistent, thereby maintaining the validity of isomorphism. After adding a valid node pair to the mapping set $M$, the label classes of the remaining unmapped nodes are updated to reflect the changes in $M$ (line~\ref{algo_hmcsp:15}).

\subsection{Upper Bound Computation and Pruning}

At each search state, \hyper{} computes an upper bound $\textit{UB}$ on the size of the largest possible mapping set that can be obtained by extending the current mapping $M$ (line~\ref{algo_hmcsp:9}). Let $\mathcal{U} = \{(U_1, U_2)\}$ be the set of current label classes. For each label class $(U_1, U_2) \in \mathcal{U}$, at most $\min(|U_1|, |U_2|)$ additional node pairs can be added to $M$. Therefore, the overall upper bound is given by:
\begin{equation}\label{eq:ub}
    \textit{UB} = |M| + \sum_{(U_1,\, U_2) \in \mathcal{U}} \min(|U_1|,\, |U_2|).
\end{equation}
If this bound is no greater than the size of the best solution found so far, i.e., $|\textit{UB}| \leq |M^*|$, the current branch is pruned and the algorithm backtracks (line~\ref{algo_hmcsp:10}).

\subsection{Heuristics}\label{subsec:heuristics}

At each search state, \hyper{} employs two heuristics to guide the search, one for selecting a label class (line~\ref{algo_hmcsp:11}) and one for selecting a node within the chosen class (line~\ref{algo_hmcsp:12}). 

\smallskip\noindent\textbf{Label Class Selection.}\label{itm:h1}
Among all label classes $(U_1, U_2) \in \mathcal{U}$, the algorithm selects the one with the smallest $\max(|U_1|, |U_2|)$. If multiple classes share the same minimum value, it selects the one that contains a node $v_1 \in U_1$ with the highest degree, where degree is defined as the total number of nodes connected within its row and column hyperedges.

\smallskip\noindent\textbf{Node Selection.}\label{itm:h2}
Within the chosen label class $(U_1, U_2)$, the algorithm selects the node $v_1 \in U_1$ with the highest degree across both row- and column-hyperedges.

\subsection{Hypergraph-Specific Components}\label{subsec:adaptations} 

We summarize the non-trivial adaptations made to the core components of \mcsplit{} to fit our hypergraph framework:

\smallskip\noindent\textbf{From Adjacency Labels to Hypergraph-Aware Labels.}
In \mcsplit{}, adjacency labels are used to encode whether an unmapped node is adjacent to any mapped nodes. However, this binary notion is not applicable in hypergraphs, where edges become higher order. In our hypergraph setting, the mapping between nodes is not determined by adjacency, but by cell values, row-hyperedge and column-hyperedge memberships simultaneously. To capture this information, we therefore introduce hypergraph-aware label $\lambda_{H}(u) = (\ell(u), \alpha(u), \beta(u))$ and maintain label classes $(U_1, U_2)$, where nodes within each class share the same label. This structure encapsulates all critical information needed to extend the current mapping $M$ in a compact form.

\smallskip\noindent\textbf{Compatibility Checks.}
In \mcsplit{}, two nodes from both graphs are considered matchable if their adjacency labels coincide. In \hyper{}, a pair $(u_1, u_2)$ is matchable only if it satisfies three domain-specific conditions: (1) cell value equality; (2) hyperedge compatibility; (3) unique hyperedge membership.

\smallskip\noindent\textbf{Label Class Refinement.}
In \mcsplit{}, when a new node pair is added, the label classes are updated based on the adjacency to mapped nodes. In \hyper{}, they are updated based on hyperedge compatibility. Specifically, after adding a pair $(v_1, v_2)$, every remaining node $u$ is regrouped based on whether its row matches $\alpha(v)$, its column matches $\beta(v)$, both, or neither. This finer-grained partitioning reduces the size of label classes, keeps the branching factor low and tightens the upper bound in practice.

\smallskip\noindent\textbf{Upper Bound Computation.}
Recall from Equation~\ref{eq:ub}, the upper bound is $UB = |M| + \sum_{(U_1,\, U_2) \in \mathcal{U}} \min(|U_1|,\, |U_2|)$. In \hyper{}, $U_1$ and $U_2$ are now defined over the hypergraph-aware labels. Due to the way our hypergraph setting partitions the label classes based on cell values and incidence for the two types of hyperedges, the resulting upper bound is empirically tighter than the one in \mcsplit{}, thereby yielding substantial pruning in practice.

\smallskip\noindent\textbf{Heuristics.}
We maintain the general approach for label class and node selection in the heuristics, but we reinterpret degree as the total number of incident nodes within a node's unique row- and column-hyperedges. The label class  continues to prioritize the class with the smallest $(|U_1|,\, |U_2|)$, which empirically leads to more effective pruning and more efficient search by minimizing the branching factor early in the search tree.

\subsection{Tolerance-Based Pruning}\label{subsec:tolerance}
% \huang{perhaps, we can divide our method into several section? there is only one subsection for our method, which looks weird.}

The branch-and-bound search tree grows exponentially in the worst case. Even after finding the optimal solution, the search may continue because some unexplored branches have upper bounds near the optimal value, which prevents pruning by the exact pruning rule. To mitigate this overhead and speed up convergence, we introduce a \emph{relative-tolerance parameter} $\delta$. Let $UB$ be the upper bound of the current branch and $|M^*|$ be the size of the best solution found so far. If $UB \leq (1 + \delta)|M^*|\text{, where } \delta \geq 0$, the branch is pruned immediately, as we consider any potential gain of $\delta|M^*|$ additional pairs negligible. This relaxes the pruning condition and accelerates the search, shrinking the size of the tree. Branches that could exceed $(1 + \delta)|M^*|$ will not be pruned, ensuring the algorithm explores all possibilities and eventually finds a solution within a $(1 + \delta)$-approximation of the optimal solution. By setting $\delta \geq 0$, \hyper{} guarantees a solution within this approximation. While $\delta = 0$ restores the exact algorithm, larger values of $\delta$ offer significant speed-ups in practice, at the cost of a bounded loss in optimality.

% *Note: Older version of this section (with Felix's comments) is moved to "archive/2_tex_sigmod/5_experiment.tex" for readability.

\section{Experiments}\label{sec:exp}

This section presents a comprehensive evaluation of our proposed method. We first describe the experimental setup (Section~\ref{subsec:exp_setup}), then analyze the method's effectiveness in identifying accurate and complete overlaps (Section~\ref{subsec:effectiveness}), and present three case studies illustrating practical applications (Section~\ref{subsec:case_studies}). We further assess computational efficiency (Section~\ref{subsec:efficiency}), scalability to large tables (Section~\ref{subsec:scalability}), and sensitivity to key parameters (Section~\ref{subsec:paramstudy}).

\subsection{Experimental Setup}\label{subsec:exp_setup}

\noindent\textbf{Datasets.}
We evaluate on two benchmark datasets, \emph{\wiki{}} and \emph{\git{}}, released by Armadillo~\cite{pugnaloni2025table}, the most recent work on table overlap size estimation. Armadillo systematically pairs tables from the Wikipedia table corpus~\cite{bleifuss2021structured, bleifuss2021secret} (the same source used by \sloth{}~\cite{zecchini2024determining}) and the GitHub repositories~\cite{hulsebos2023gittables}, producing a uniform distribution of overlap sizes that enables unbiased comparison across varying overlap complexities. Although designed for estimating overlap sizes rather than finding overlaps, the datasets provide ground truth for the largest rectangular overlaps across structurally diverse tables. Defined under rectangular constraints, these annotations remain suitable for evaluating shape-agnostic overlap, as they serve as a lower bound for comparison. 

\emph{\wiki{}} and~\emph{\git{}} contain 128k and 256k tables, respectively, with sizes ranging from 2 to \numprint{99856} cells, and their distributions are shown in Figure~\ref{fig:table_stats}. We also report hypergraph statistics in Figure~\ref{fig:hypergraph_stats}. Overall, \git{} has larger and heavier-tailed node and row-hyperedge counts than \wiki{}, while column-hyperedge counts remain relatively small in both datasets. For evaluation, we uniformly sample 500k table pairs from each dataset following Armadillo’s pairing distribution, as our task does not require cross-validation.

\begin{figure}[t]
    \centering
    \subfloat[\wiki{}]{
        \includegraphics[width=0.23\textwidth]{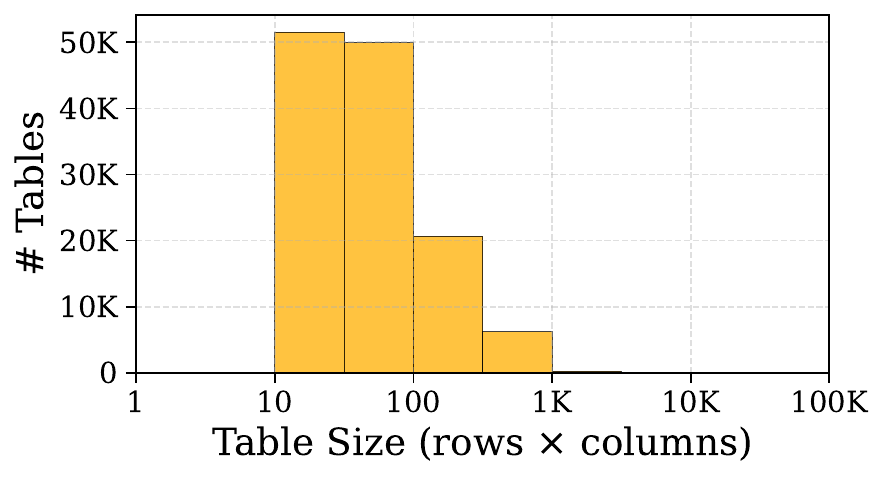}
    }
    \subfloat[\git{}]{
        \includegraphics[width=0.23\textwidth]{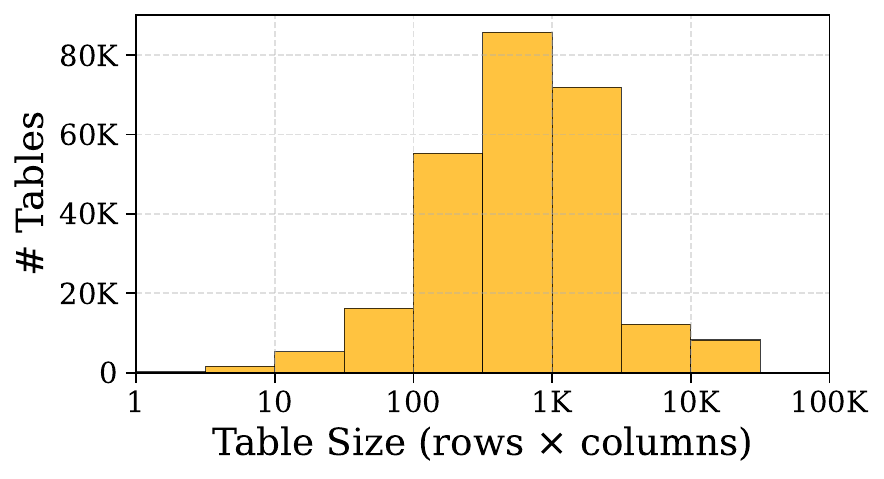}
    }
    \caption[]{Distributions of table sizes (log-scale)}\label{fig:table_stats}
    % \vspace{-1em}
\end{figure}
\begin{figure}[t]
    \centering
    \subfloat[\wiki{}]{
        \includegraphics[width=0.48\textwidth]{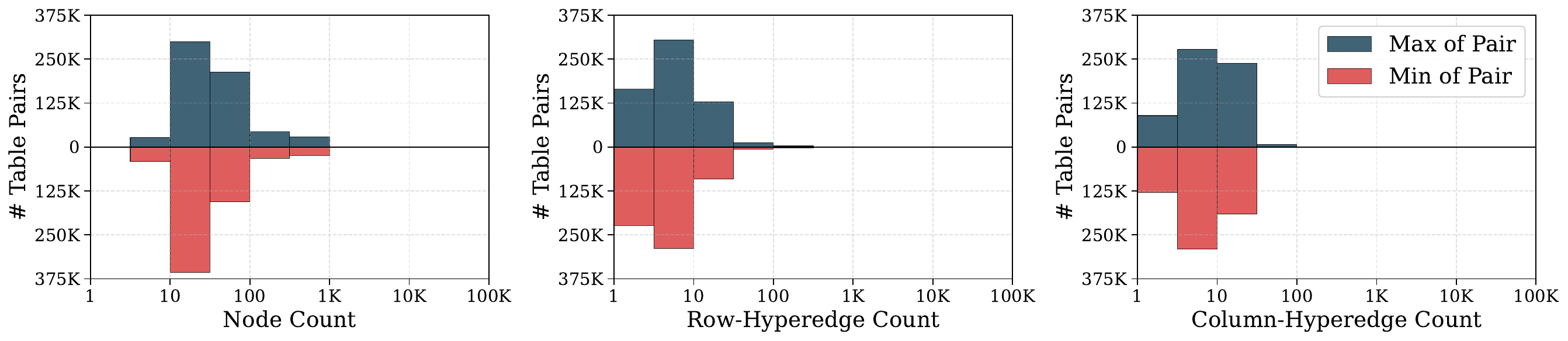}
    }
    \vfill
    \vspace{-1em}
    \subfloat[\git{}]{
        \includegraphics[width=0.48\textwidth]{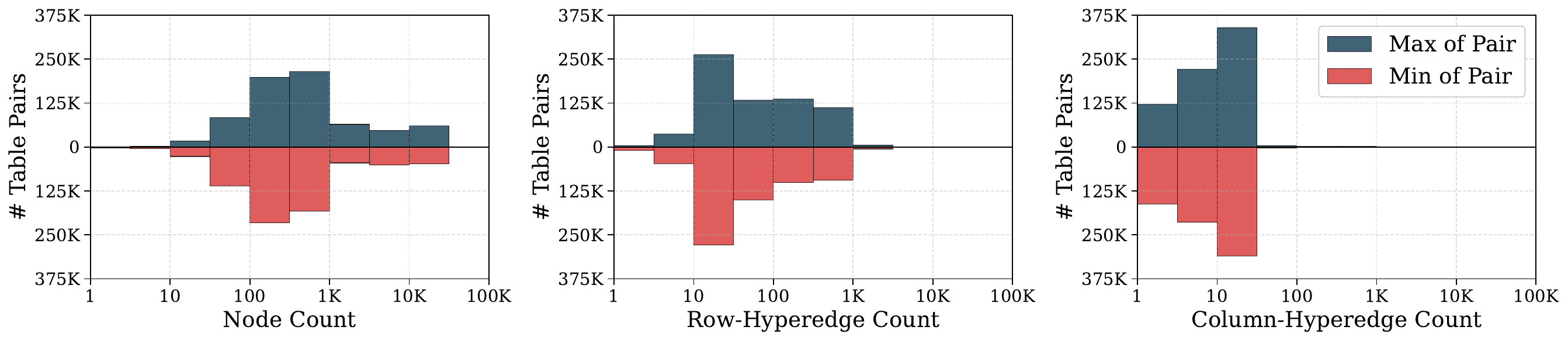}
    }
    \caption[]{Distributions of node counts, row-hyperedge counts and column-hyperedge counts (log-scale)}\label{fig:hypergraph_stats}
\end{figure}

\smallskip\noindent\textbf{Methods for Comparison.} 
We compare \hyper{} against three baselines: one for the rectangular overlap problem and two addressing \salto{} via similarity and LLM approaches.

\begin{itemize}[leftmargin=*]
    \item \sloth{}~\cite{zecchini2024determining}: A state-of-the-art exact solver for the largest rectangular overlap problem.

    \item \crjaccard{}: A similarity-based heuristic that aligns rows and columns using greedy Jaccard similarity. It first computes pairwise Jaccard similarities between all columns, matching each column in one table to the most similar unmatched column in the other. Using the aligned columns, it then performs greedy row matching by computing pairwise Jaccard similarities over corresponding column values. The final overlap is obtained through cell-wise comparison over the matched rows and columns, producing a non-rectangular overlap region.

    \item \gpt{}: A large language model (LLM)-based approach that serializes each table into a structured text format and prompts the model to infer the row and column permutations that maximize \salto{}. We use deterministic decoding (\emph{temperature}=0) and issue a single prompt per table pair. Prompt template and details are provided in the supplementary material (Appendix~\ref{app:prompt}).

    \item \hyper{}: Our proposed exact method for solving \salto{}. Unless otherwise stated, we run the exact configuration (i.e., tolerance-based pruning disabled) by setting $\delta=0$ (Section~\ref{subsec:tolerance}). We only enable it in the parameter study (Section~\ref{subsec:paramstudy}).
\end{itemize}

\noindent
\crjaccard{} and \gpt{} serve as heuristic baselines that approximate \salto{} through similarity-based alignment and LLM inference, respectively. In contrast, \sloth{} provides an exact solution to the more restrictive rectangular variant of the \salto{} problem. Other potential approaches---such as Jaccard variants (standard, bag-based, size-normalized)~\cite{pugnaloni2025table}, Armadillo~\cite{pugnaloni2025table} for overlap size estimation, and techniques like schema matching~\cite{liu2025magneto, aumueller2005schema, hong2002coma}, entity resolution~\cite{wang2023sudowoodo, wu2020zeroer}, or unionable table discovery~\cite{nargesian2018table, khatiwada2023santos, fan2023semantics}---are incompatible with our task, as they either estimate only the extent of overlap or perform schema- or entity-level alignment rather than identifying exact overlapping cells. More details are discussed in Sections~\ref{sec:intro} and~\ref{sec:related_work}.
% Astute readers may find several seemingly feasible competitors, such as standard Jaccard similarity, Jaccard under bag semantics, its variants normalized by the smaller table size, and Armadillo, introduced in~\cite{pugnaloni2025table}, that estimate overlap ratios. However, these methods estimate only the extent of overlap, without identifying the actual overlapping cells, and hence are incompatible with our method.  This distinction is critical for two reasons. \todo{First, our problem requires verifiable estimations: overlap ratio alone does not ensure that the estimated overlap corresponds to any valid row and column permutation or meaningful structural match. Second, these methods are incompatible with our task because it demands exact, interpretable overlap detection based on aligned table content. }Accurate localization of overlap is essential not only for verification but also for downstream applications that rely on the exact overlap positions.

\smallskip\noindent\textbf{Environment}. 
All experiments were conducted on a Linux server with Intel Xeon~E5 CPUs and 512\,GB of RAM. The code, implemented in Python, is available at~\cite{hypersplitcodeanon}.

\subsection{Effectiveness Analysis}\label{subsec:effectiveness}

\begin{figure}[t]
    \centering
    \vspace{-1em}
    \subfloat[\hyper{} vs. \sloth{}]{
        \includegraphics[width=0.43\textwidth]{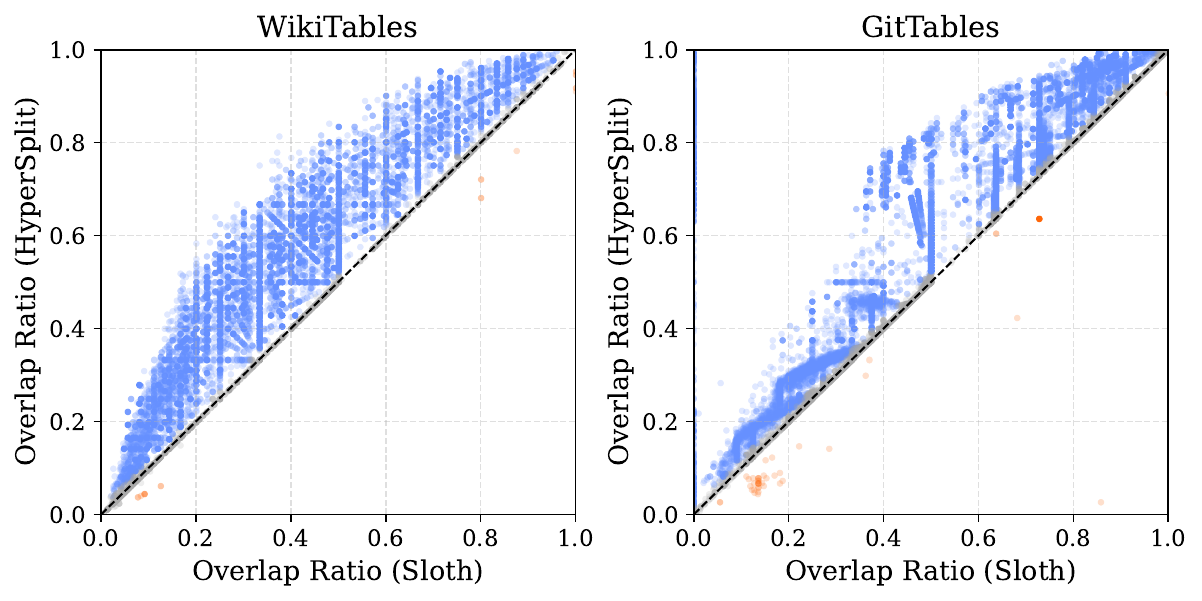}
    }
    \vfill
    \vspace{-1em}
    \subfloat[\hyper{} vs. \crjaccard{}]{
        \includegraphics[width=0.43\textwidth]{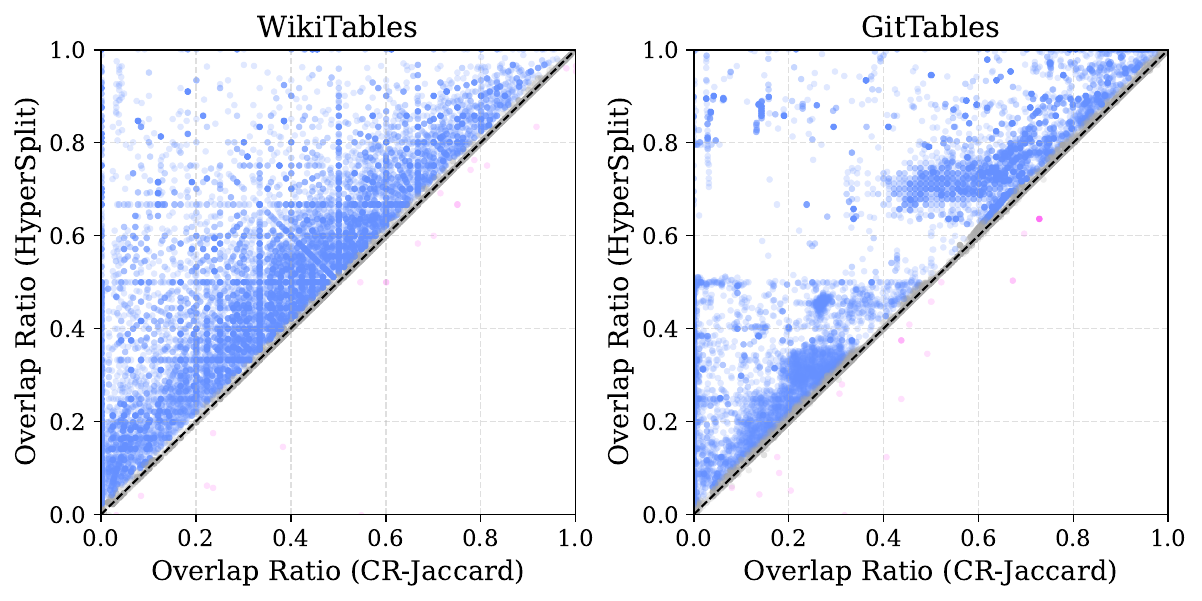}
    }
    \vfill
    \vspace{-1em}
    \subfloat[\hyper{} vs. \gpt{}]{
        \includegraphics[width=0.43\textwidth]{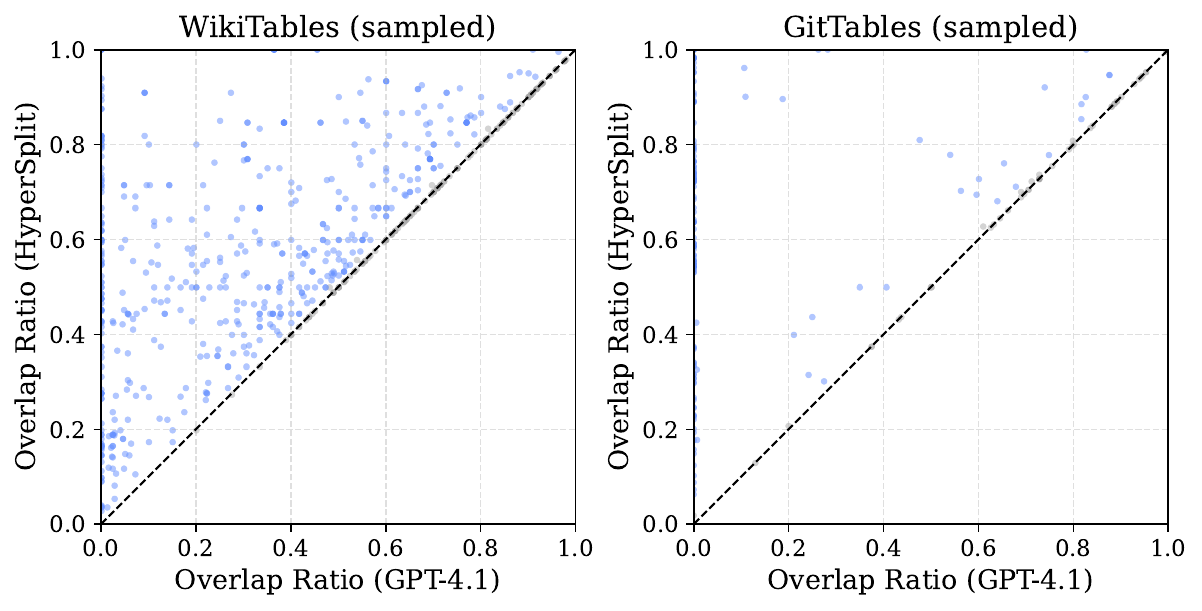}
    }
    \caption[]{Overlap ratio comparison of \hyper{} against \sloth{}, \gpt{} and \crjaccard{}}\label{fig:effectiveness}
    % \vspace{-1em}
\end{figure}

We evaluate the effectiveness of \hyper{} in identifying overlaps between table pairs. The comparison includes two baselines for the \salto{} problem, \crjaccard{} and \gpt{}, and one for the rectangular overlap problem, \sloth{}, to contrast the two formulations. The evaluation is based on the \emph{overlap ratio}, defined as the fraction of overlapping cells relative to the smaller table’s area. This ratio provides an interpretable and size-normalized measure of overlap completeness that remains comparable across tables of different scales. Figure~\ref{fig:effectiveness} plots the ratios of overlap obtained by each baseline (x-axis) against those obtained by \hyper{} (y-axis) on \wiki{} and \git{}. For \gpt{}, we report results on \numprint{1000} randomly sampled pairs due to its high inference time, as further discussed in Section~\ref{subsec:efficiency}. Points above the diagonal indicate pairs where \hyper{} identifies a larger overlap.

Across both datasets, almost all points lie above the diagonal, showing that \hyper{} consistently finds larger overlaps. The advantage is especially clear over \sloth{}, where many points cluster near the x-axis---cases where rectangular constraints fail to capture irregular or non-contiguous matches that \hyper{}’s shape-agnostic formulation successfully recovers. Compared with \crjaccard{}, which performs greedy top-1 row and column matching based on Jaccard similarity, \hyper{} achieves higher overlap ratios by avoiding locally optimal but globally inconsistent alignments. For \gpt{}, the performance gap reflects the inherent combinatorial complexity of reasoning over large tabular search spaces, where LLM inference struggles to maintain consistent structural alignment across rows and columns.

Since \hyper{} performs exact search over cell mappings, it should always find at least as large an overlap as any baseline. The few points below the diagonal correspond to cases where the search was terminated early due to time limits (60 seconds), returning the largest overlap found within the allocated runtime. In summary, \hyper{} captures a complete and consistent overlaps beyond rectangular or heuristic matching, demonstrating its strength in uncovering valid cell-level matches between tables.

\subsection{Case Studies}\label{subsec:case_studies}

We conduct three case studies to illustrate the practical impact of \hyper{} in real-world scenarios. Across all case studies, we compare \hyper{} against \sloth{}, which does compute exact overlaps but is restricted to rectangular regions, making it the most relevant baseline for highlighting the advantages of our shape-agnostic formulation.
% As mentioned in Section~\ref{sec:intro}, a natural alternative for discovering overlap is through schema matching, entity resolution, or unionable table discovery. These methods, however, rely on schema or metadata and cannot localize exact overlapping cells. In contrast, \hyper{} operates directly on raw tables, requires no prior knowledge of table or attribute semantics, and recovers the actual, exact overlap at the cell level. Such capabilities are essential in applications where metadata is incomplete or inconsistent, and where downstream tasks demand exact overlap localization, such as detecting potential copying across sources (Case Study~1), identifying redundant table content (Case Study~2), and comparing different table versions to localize shared or edited regions (Case Study~3). Throughout all case studies, we compare against \sloth{}, which does compute exact overlaps but is restricted to rectangular regions, making it the most relevant baseline for highlighting the advantages of our shape-agnostic formulation.

\begin{table}[t]
    \centering
    \small
    \setlength{\tabcolsep}{4pt}
    \begin{threeparttable}
        \caption{Clustering quality across thresholds $\gamma$}\label{tab:cluster_quality}
        \begin{tabular}{cccc ccc ccc}
            \toprule
            \multirow{2}{*}{\textbf{$\gamma$}} &
            \multicolumn{3}{c}{\textbf{\jaccard{}}} &
            \multicolumn{3}{c}{\textbf{\sloth{}}} &
            \multicolumn{3}{c}{\textbf{\hyper{}}} \\
            \cmidrule(lr){2-4} \cmidrule(lr){5-7} \cmidrule(lr){8-10}
             & \textbf{ARI} & \textbf{NMI} & \textbf{\#\,ST} 
             & \textbf{ARI} & \textbf{NMI} & \textbf{\#\,ST} 
             & \textbf{ARI} & \textbf{NMI} & \textbf{\#\,ST} \\
            \midrule
            0.81 & 0.93 & 0.89 & 14 & 0.89 & 0.88 & 13 & \textbf{1.00} & \textbf{1.00} & \textbf{15} \\
            0.72 & 0.83 & 0.81 & 14 & 0.90 & 0.91 & 15 & \textbf{1.00} & \textbf{1.00} & \textbf{17} \\
            0.64 & 0.56 & 0.63 & 14 & 0.70 & 0.73 & 17 & \textbf{1.00} & \textbf{1.00} & \textbf{23} \\
            0.56 & $-0.02$ & 0.32 & 16 & 0.54 & 0.75 & 40 & \textbf{1.00} & \textbf{1.00} & \textbf{42} \\
            0.49 & $-0.02$ & 0.27 & 21 & 0.70 & 0.72 & 46 & \textbf{1.00} & \textbf{1.00} & \textbf{47} \\
            \bottomrule
        \end{tabular}
        \begin{tablenotes}[flushleft]
            \footnotesize
            \item \textbf{ARI} = Adjusted Rand Index; \textbf{NMI} = Normalized Mutual Information
            \item \textbf{\#\,ST} = number of similar tables detected across all clusters
        \end{tablenotes}
    \end{threeparttable}
    % \vspace{-1em}
\end{table}

\subsubsection{Case Study 1 -- Cross-Source Copy Detection} 

Detecting potential copying of tables across web sources is a key problem in truth discovery~\cite{li2012truth, li2015scaling, dong2009truth, dong2010global}. As copying typically manifests as high content similarity, we detect similar tables directly from raw tables without any manual preprocessing such as schema alignment or reliance on metadata. This setting is more challenging, as similar tables must be identified purely from structural and cell-level content without schema guidance.

We use the \emph{Stock} dataset~\cite{li2012truth}, which contains daily stock tables collected from 55~sources over 21~days. Following prior work~\cite{li2012truth, zecchini2024determining}, two tables are considered similar if their similarity exceeds a given threshold~$\gamma$. For \hyper{}, similarity is the \emph{overlap ratio}, i.e., the number of overlapping cells normalized by the smaller table’s total cell count. For \sloth{}, which restricts overlaps to rectangular regions, similarity is the product of the row and column overlap ratios: $\gamma=0.81$ corresponds to $0.9$ row $\times$ $0.9$ column overlap, while $\gamma=0.72$, $0.64$, $0.56$, and $0.49$ correspond to $(0.8, 0.9)$, $(0.8, 0.8)$, $(0.7, 0.8)$, and $(0.7,0.7)$ row–column overlaps, respectively. Following previous studies, we also use \jaccard{}, where each table is flattened into a set of cells values and compared using the standard Jaccard index. Clusters of similar tables are then formed by grouping all tables whose pairwise similarities exceed the threshold~$\gamma$.

Table~\ref{tab:cluster_quality} reports the number of detected similar tables (\#\,ST) and cluster quality, measured by the \emph{Adjusted Rand Index} (ARI)~\cite{hubert1985comparing} and \emph{Normalized Mutual Information} (NMI)~\cite{danon2005comparing}, two widely adopted metrics for evaluating clustering quality, where higher values indicate better agreement with the ground truth. Since \hyper{} provides exact and complete cell-level overlap, its clusters are used as the reference ground truth for comparing the other methods. Across all thresholds, \hyper{} detects more similar tables than both \sloth{} and \jaccard{}. \sloth{}’s rectangular constraint misses irregular or non-contiguous overlaps, yielding fragmented clusters and lower ARI and NMI. \jaccard{}, which flattens tables under set semantics and ignores repetition and positional alignment of cell values, deteriorates rapidly as $\gamma$ decreases. These results highlight that exact, cell-level overlap is critical for accurate cluster formation.

Figure~\ref{fig:cs1} illustrates this behavior. While \sloth{} and \jaccard{} split cohesive groups into smaller, disjoint clusters, \hyper{} groups tables that share the same underlying content. Together with Table~\ref{tab:cluster_quality}, these findings show that \hyper{}'s \salto{} captures more complete and content-consistent clusters across sources.

\begin{figure}[t]
    \centering
    \includegraphics[width=0.48\textwidth]{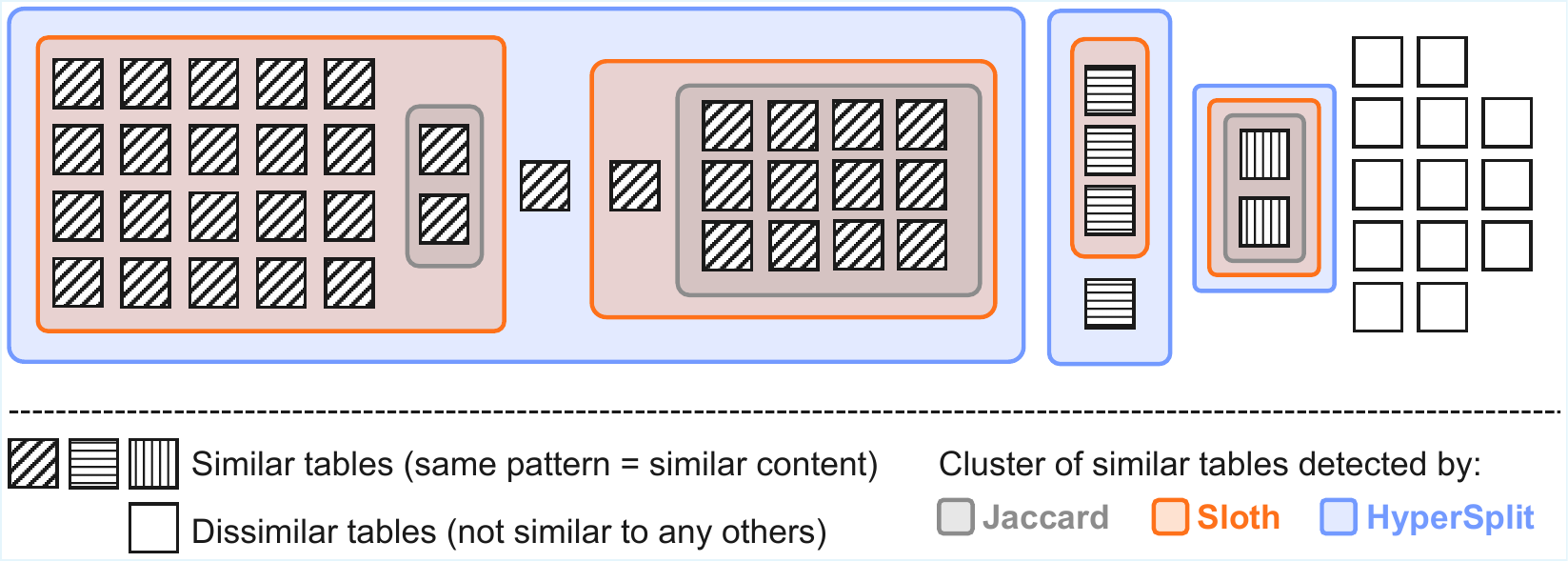}
    \caption{Clusters of similar tables detected by Jaccard, \sloth{} and \hyper{} at \boldsymbol{$\gamma=0.56$}. Each square represents a table, patterned squares indicate tables sharing similar content.}\label{fig:cs1}
\end{figure}

\subsubsection{Case Study 2 -- Data Deduplication} 

\begin{figure*}[t]
    \centering
    \includegraphics[width=0.98\textwidth]{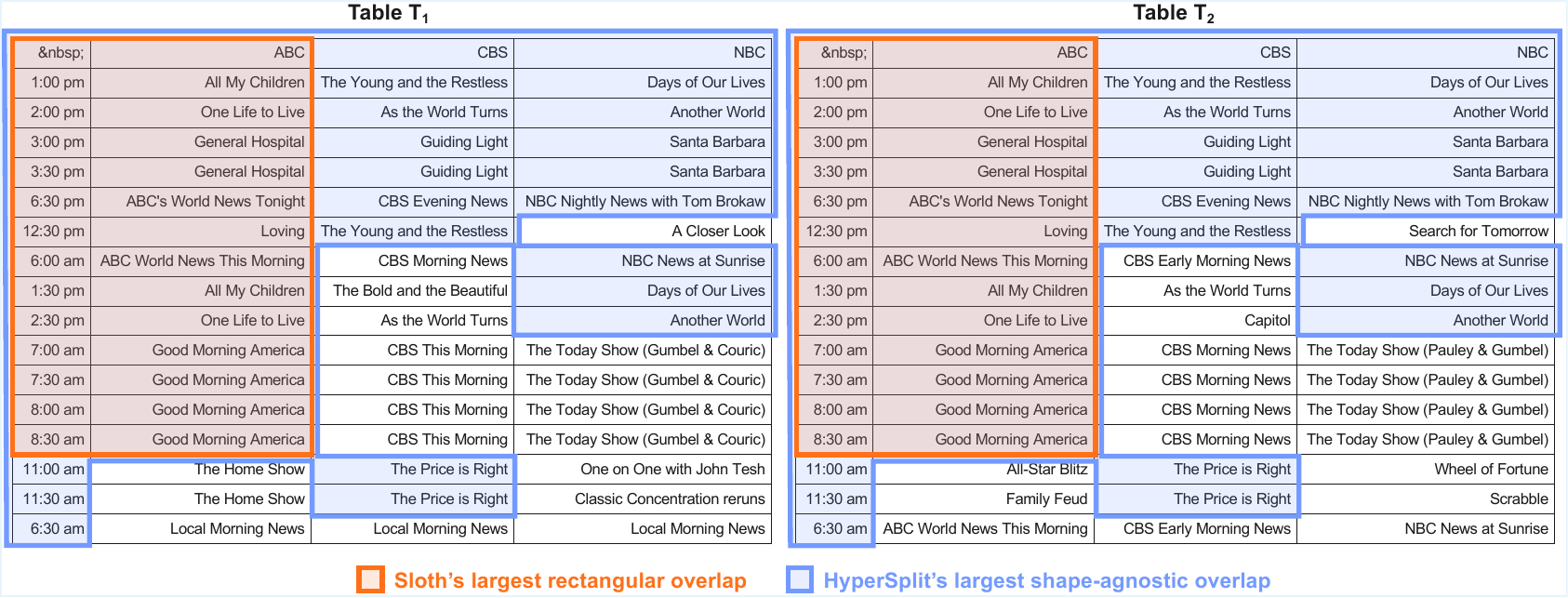}
    \caption{Comparison between \sloth{}’s largest rectangular overlap and our \hyper{}'s \salto{} between two tables}\label{fig:deduplication}
    % \vspace{-1em}
\end{figure*}

We next study redundant content deduplication within tables in large repositories~\cite{koch2023duplicate, chu2016data, ilyas2015trends}, such as Wikipedia, where overlapping content often arises from copy-paste, template reuse, or partial revisions across pages~\cite{bleifuss2021secret, bhagavatula2015tabel}. Localizing repeated cells supports repository maintenance by identifying tables that duplicate curated content, reducing inconsistent updates when the same facts are edited in one place but not others, and guiding consolidation into a single source of truth. Unlike entity resolution or schema matching, which target fuzzy record-level duplication or semantic attribute alignment, our goal here is to localize the exact shared cells between two tables and thereby expose redundant regions that can be consolidated, removed, or tracked to propagate edits consistently across duplicated copies.

We analyze a \wiki{} table pair with shared content (Figure~\ref{fig:deduplication}). By matching identical cell values and enforcing a single globally consistent row/column alignment, \hyper{} recovers a large shared region of 49 cells (blue). In contrast, \sloth{} detects only 28 cells (orange) because its rectangular constraint excludes valid matches outside the largest contiguous block. This example illustrates that shape-agnostic overlap can capture a more complete redundant region even when the shared content is distributed across multiple aligned row/column pairs. This finer-grained localization supports redundancy reduction and can provide evidence of derivation or reuse between tables.

Our formulation uses exact equality to identify overlaps. In cases with heterogeneous surface forms, such as unit differences or abbreviation and lexical variants, one can normalize or canonicalize values before constructing the hypergraphs. Supporting fuzzy or learned matching would instead require redefining the overlap objective, for example with thresholded or weighted matches -- we leave this direction outside the scope of this paper. 
Another consideration is homonymous values, where identical strings have different meanings and may lead to spurious matches under equality. A practical workaround is to use column context, such as headers, units, or types to restrict matches to compatible columns. This context can be taken from metadata when available or inferred from column values. Finally, if the overlap contains small isolated components, meaning tiny disconnected groups of matched cells that may look spurious or be undesirable for a downstream task, one can post-filter the result by retaining only connected components above a minimum size.

\subsubsection{Case Study 3 -- Version Comparison} 

\begin{figure}[t]
    \centering
    \includegraphics[width=0.47\textwidth]{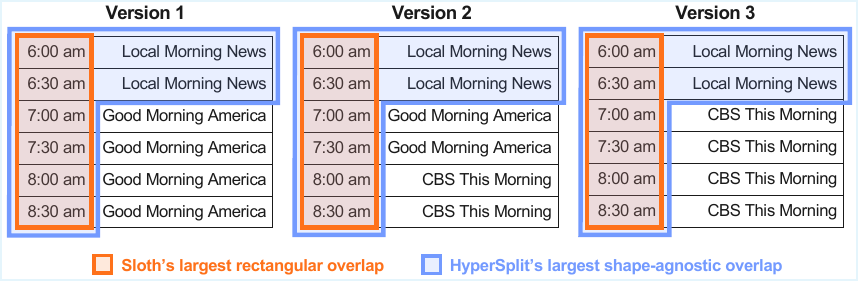}
    \caption{Comparison between \sloth{}’s largest rectangular overlap and \hyper{}'s \salto{} across three table versions}\label{fig:version_compare}
\end{figure}

We extend the analysis to version comparison~\cite{bleifuss2018exploring}, where multiple versions of the same table evolve over time. In this setting, overlap detection helps localize edits, quantify changes, and identify core regions preserved across versions. We analyze tables from \wiki{} representing different versions of the same entry and apply both \sloth{} and \hyper{} to compute pairwise overlaps and determine the largest shared region across all versions.

As shown in Figure~\ref{fig:version_compare}, we illustrate three versions of a table to demonstrate the observed behavior. Several cells (``Good Morning America'' and ``CBS This Morning'') are not simultaneously shared across all three versions, even though each pair does include these values in their pairwise maximal overlap. This prevents the overlap from forming a single contiguous rectangle. Since \sloth{} enforces rectangular overlap, such constraint prevents it from identifying additional shared cells (``Local Morning News'') in the affected column. In contrast, \hyper{} captures all shared cells, including those that become non-contiguous after local edits, yielding a more complete and meaningful overlap. Similar patterns are observed across other versioned tables where local edits disrupt rectangular alignment. The ability is crucial for table evolution scenarios: as tables grow and undergo edits, such as insertions, deletions, or modifications, rectangular alignment is often disrupted, causing \sloth{} to miss valid overlaps. By preserving shared content despite such changes, \hyper{} enables reliable version tracking, change localization, and update propagation across evolving tables.

This setting also connects to semantic data versioning systems, such as Explain-Da-V~\cite{shraga2023explaining}, which characterize changes at the level of attributes and entities. \salto{} is complementary: it provides exact, permutation-invariant cell overlaps that can serve as a low-level primitive to identify the unchanged cell regions between two versions, which can then be fed into Explain-Da-V for higher-level explanations. Applying \salto{} in this workflow also raises challenges, such as handling heterogeneous value representations beyond exact equality and scaling to long version histories where many pairwise overlaps would need to be computed.

\subsection{Efficiency Analysis}\label{subsec:efficiency}

\begin{table}[t]
    \small
    \centering
    % \caption{Mean $\pm$ standard deviation of runtime (s) across table pairs for each method}
    \caption{Runtime per table pair on \wiki{} and \git{} in seconds (mean \boldsymbol{$\pm$} std). For \hyper{}, \textsf{LabelClassSel} and \textsf{NodeSel} are the two selection heuristics; ``w/o'' removes the indicated one (or both).}
    \label{tab:efficiency}
    \begin{tabular}{lcc}
        \toprule
        \textbf{Method} & \textbf{\wiki{}} & \textbf{\git{}} \\
        \midrule
        \sloth{} & 0.856 $\pm$ 4.700 & 0.180 $\pm$ 2.045 \\  
        \crjaccard{}  & 0.001 $\pm$ 0.002 & 0.162 $\pm$ 0.446 \\
        \gpt{}        & 2.551 $\pm$ 3.685 & 13.537 $\pm$ 39.323 \\
        \hyper{}      & 0.358 $\pm$ 2.997 & 2.702 $\pm$ 9.091 \\
        \hyper{} (w/o \textsf{LabelClassSel}) & 8.344 $\pm$ 20.284 & 38.181 $\pm$ 27.970 \\  
        \hyper{} (w/o \textsf{NodeSel}) & 6.590 $\pm$ 18.364 & 29.978 $\pm$ 29.379 \\ 
        \hyper{} (w/o both) & 8.376 $\pm$ 20.305 & 38.388 $\pm$ 27.879 \\
        \bottomrule
    \end{tabular}
    % \vspace{-1em}
\end{table}

We analyze the runtime of \sloth{}, \crjaccard{}, \gpt{}, and \hyper{} on the \salto{} problem. As shown in Table~\ref{tab:efficiency}, \crjaccard{} is the fastest because it performs heuristic greedy matching based solely on Jaccard similarity, without exploring the combinatorial search space. In contrast, \gpt{} is substantially slower, approximately $5$--$7\times$ slower than \hyper{} on average, due to token-level reasoning and the need to sequentially process each table pair through large-model inference, making it computationally expensive for large-scale evaluation. 

\sloth{} is much faster than \hyper{}, as expected, because finding the largest rectangular overlap admits a more constrained search space than recovering the largest shape-agnostic overlap. Interestingly, \sloth{} is faster on \git{} than on \wiki{}: \sloth{}’s cost is dominated by the number of seeds, i.e., column pairs whose value sets intersect, rather than by table size. Tables in \wiki{} exhibit far more such intersections, producing a larger number of seeds and a larger search space, whereas tables in \git{} yield fewer seeds and thus terminate faster. This difference reflects the efficiency--expressiveness tradeoff between rectangular and shape-agnostic overlap definitions.

\hyper{} provides a practical balance between computational cost and completeness. While performing exact cell-level search, it remains efficient, averaging through effective pruning heuristics and hypergraph-aware candidate selection, which significantly reduce the search space. This is further supported by the ablation results: disabling either \textsf{LabelClassSel} or \textsf{NodeSel} heuristic (Section~\ref{subsec:heuristics}) increases runtime by $11$--$23\times$, and \textsf{LabelClassSel} contributes the largest savings. Taken together, \hyper{} achieves practical efficiency while remaining exact.

We further examine \hyper{}'s efficiency under challenging cases by tracking its \emph{runtime-overlap progression} on large table pairs. For this analysis, we select a set of complex pairs from \git{} where the tables exceed \numprint{10000} cells. Figure~\ref{fig:efficiency} plots the mean overlap ratio achieved by \hyper{} over time, with the shaded region indicating one standard deviation across selected pairs. \hyper{} exhibits a clear convergence trend. The overlap ratio increases rapidly within the first few seconds and stabilizes after roughly 50\,s. This behavior shows that \hyper{} effectively expands the overlap efficiently while keeping computation manageable through effective pruning of unpromising branches as the search progresses.

\begin{figure}[t]
    \centering
    \includegraphics[width=0.44\textwidth]{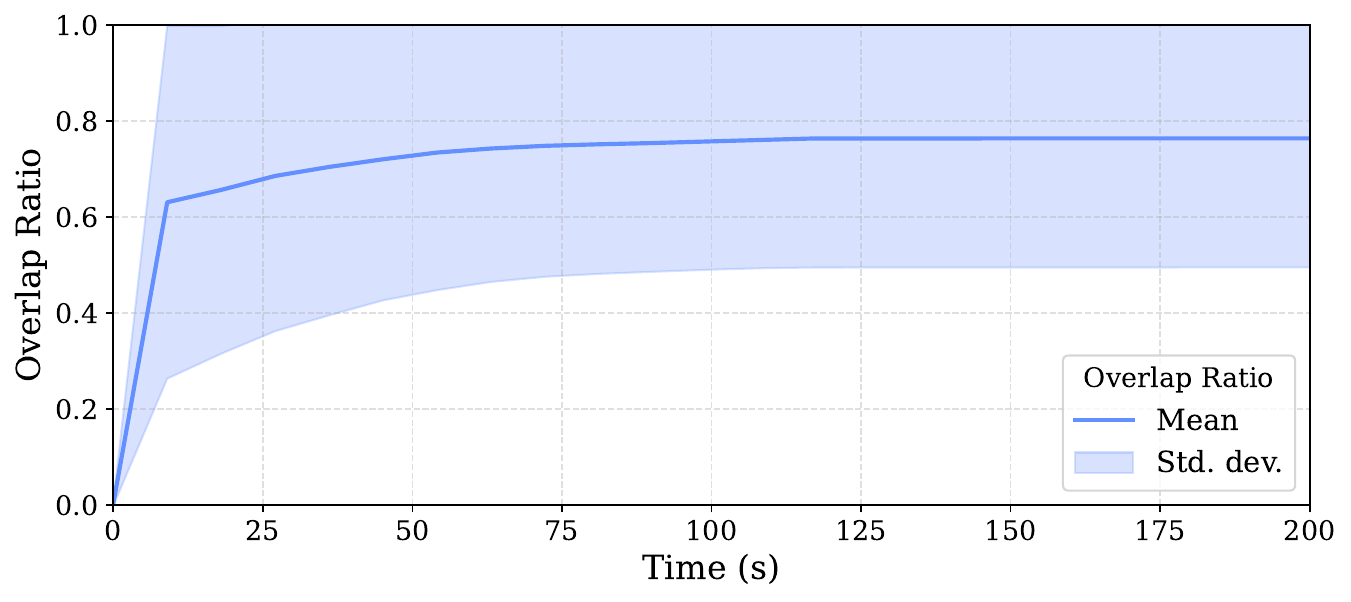}
    \label{subfig:scala1}
    \caption[]{Runtime-overlap progression on large table pairs on \git{} dataset}\label{fig:efficiency}
\end{figure}

\subsection{Scalability Analysis}\label{subsec:scalability}

\begin{figure}[t]
    \centering
    \subfloat[Varying Overlap Ratio]{
        \includegraphics[width=0.23\textwidth]{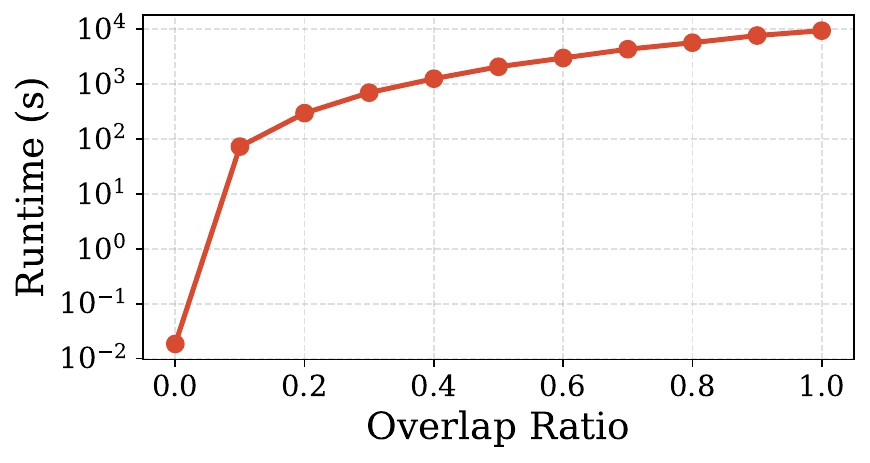}
        \label{subfig:vary_olap}
    }
    \subfloat[Varying Table Size]{
        \includegraphics[width=0.23\textwidth]{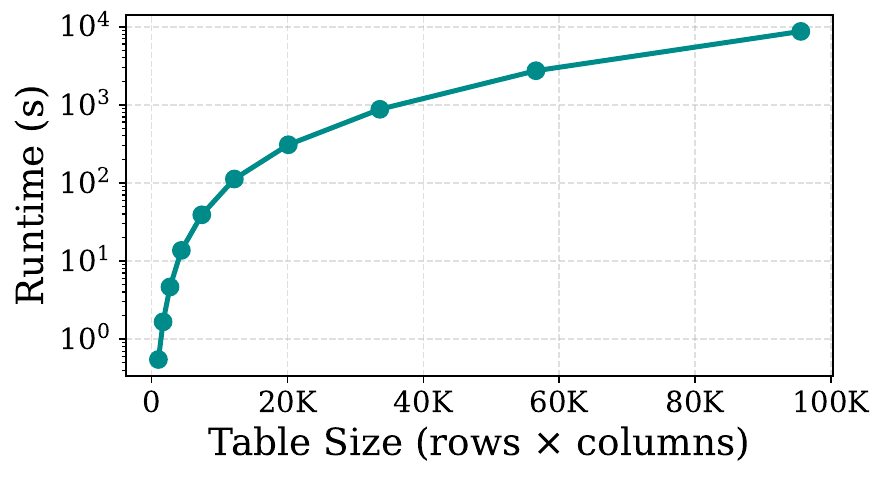}
        \label{subfig:vary_size}
    }
    % \vspace{0.5em}
    \caption[]{Scalability Study}\label{fig:scalability}
\end{figure}

Scalability is a critical yet challenging aspect of table overlap computation. As table dimensions grow or the overlap between them increases, the number of potential row and column alignments expands combinatorially, resulting in a combinatorially large search space. We evaluate the scalability of \hyper{} using synthetic test cases that systematically increase task difficulty along two controlled dimensions: the overlap ratio and the table size.

\smallskip\noindent\textbf{Varying Overlap Ratio.} We fix both tables to 1000~rows and 50~columns, and vary the overlap ratio from 0 to 1 in increments of 0.1, thereby increasing the proportion of overlapping cells while keeping the dimensions constant. Figure~\ref{subfig:vary_olap} shows that runtime increases with the overlap ratio, as higher overlap introduces more candidate mappings to explore. The growth remains well controlled across the entire range, showing that \hyper{} effectively prunes redundant mappings through its hypergraph-aware structure and bidomain-guided search.

\smallskip\noindent\textbf{Varying Table Size.} We fix the overlap ratio at 0.5 and progressively enlarge both tables by 30\% per step, from 100~rows and 10~columns up to $1050{\times}91$. This jointly increases the number of rows and columns, expanding the total number of cells, which is the dominant factor in computational complexity. Figure~\ref{subfig:vary_size} shows that the runtime follows a smooth logarithmic trend relative to table size, indicating that \hyper{} scales predictably and remains efficient as input size increases by more than an order of magnitude.

\smallskip Overall, these results demonstrate that \hyper{} achieves robust and predictable scalability under increasing overlap complexity and table size, validating its effectiveness for large-scale and highly overlapping real-world tables.

\begin{table}[t]
    \centering
    \small
    \begin{threeparttable}[H]
        \caption{Impact of pruning tolerance $\delta$ and depth $n$ on mean overlap ratio and runtime across table pairs}
        \label{tab:delta_pruning}
        \begin{tabular}{cc cc cc}
            \toprule
            \multirow{2}{*}{$\boldsymbol{\delta}$} & \multirow{2}{*}{$\boldsymbol{n}$} 
            & \multicolumn{2}{c}{\textbf{\wiki{}}} 
            & \multicolumn{2}{c}{\textbf{\git{}}} \\
            \cmidrule(lr){3-4} \cmidrule(lr){5-6}
            & & \textbf{Overlap Ratio} & \textbf{Time (s)} 
              & \textbf{Overlap Ratio} & \textbf{Time (s)} \\
            \midrule
            -- & -- & 0.63 & 0.36 & 0.69 & 2.70 \\
            \midrule
            \multirow{3}{*}{0.5} 
                & 10  & 0.62 & 0.24 & 0.69 & 2.15 \\
                & 20  & 0.61 & 0.18 & 0.69 & 1.89 \\
                & all & 0.56 & 0.12 & 0.68 & 1.71 \\
            \midrule
            \multirow{3}{*}{1.0} 
                & 10  & 0.61 & 0.24 & 0.69 & 1.95 \\
                & 20  & 0.59 & 0.18 & 0.68 & 1.78 \\
                & all & 0.55 & 0.12 & 0.68 & 1.64 \\
            \midrule
            \multirow{3}{*}{2.0} 
                & 10  & 0.61 & 0.18 & 0.69 & 1.85 \\
                & 20  & 0.57 & 0.18 & 0.68 & 1.68 \\
                & all & 0.44 & 0.12 & 0.62 & 1.48 \\
            \bottomrule
        \end{tabular}
        \begin{tablenotes}[flushleft]
          \footnotesize
          \item all = pruning is applied at every level
        \end{tablenotes}
    \end{threeparttable}
    % \vspace{-1em}
\end{table}

\subsection{Parameter Study}\label{subsec:paramstudy}

We analyze the impact of the pruning tolerance parameter $\delta$ and pruning depth parameter $n$ in \hyper{}'s pruning strategy. To balance efficiency and approximation quality, pruning is applied only during the first $n$ branch levels, where early-stage pruning eliminates large unpromising subtrees while deeper branches are explored more conservatively to preserve accuracy.

Table~\ref{tab:delta_pruning} reports the mean overlap ratio and runtime across table pairs under different pruning configurations. We vary $\delta \in \{0.5, 1.0, 2.0\}$ and $n \in \{10, 20\}$, and also include a configuration where pruning is applied at all levels ($n = \text{all}$). The baseline without pruning achieves the highest overlap ratios but also the longest runtimes. Moderate pruning with $\delta = 1.0$ and $n = 10$, achieves a favorable trade-off between efficiency and accuracy. On \wiki{}, runtime decreases by 33\%, from 0.36\,s to 0.24\,s, while maintaining a high mean overlap ratio around 0.61. Similarly, on \git{}, pruning reduces runtime from 2.70\,s to 1.95\,s with no change in overlap ratio (0.69). Increasing $n$ or $\delta$ further accelerates the search but introduces more noticeable degradation in overlap ratio.

Collectively, these results show that tolerance-based pruning effectively improves computational efficiency while preserving solution quality. Early-stage pruning within the first 10–20 levels provides the best balance, offering significant runtime reduction without substantial loss of overlap accuracy, especially on large and complex tables.
\section{Related Work}\label{sec:related_work}

We organize related work into five areas: (i)~table overlap, (ii)~data discovery, (iii)~data integration and matching, (iv)~maximum common subgraph, and (v)~hypergraph-based modeling for tables.

\smallskip\noindent\textbf{Table Overlap.}
The most closely related work is \sloth{}~\cite{zecchini2024determining}, which finds the largest overlap between two tables by permuting rows and columns but constrains the overlap to be a \emph{contiguous rectangular region}. In contrast, our formulation generalizes to \emph{arbitrary, non-contiguous shapes}, capturing overlaps beyond rectangular alignment. Armadillo~\cite{pugnaloni2025table} extends \sloth{}’s rectangular formulation to \emph{estimate overlap sizes} without materializing the overlapping cells, which is useful for coarse filtering in large-scale retrieval or blocking. CENTS~\cite{xiao2025cents} instead selects informative table regions via similarity-based scoring for LLM understanding, prioritizing semantic salience rather than exact shared content. Thus, \salto{} complements these lines of work by recovering the \emph{exact} shared cells needed in tasks such as identifying copied or unchanged regions. Other studies detect duplicate tables through overlap signals. \cite{koch2023duplicate} applies the XASH hash function~\cite{esmailoghli2022mate} to identify tables sharing identical tuples under column permutations, which captures only perfect table duplicates or full row containment. \cite{bleifuss2021structured} matches tables across multiple Wikipedia page revisions using Ruzicka similarity to track table evolution. These approaches focus on global similarity or duplication but not partial or shape-agnostic overlap.

\smallskip\noindent\textbf{Data Discovery.}
Data discovery in data lakes~\cite{miller2018open, das2012finding, zhang2020finding, bogatu2020dataset} is typically divided into unionable or joinable table search. For unionable search, \cite{nargesian2018table} first introduces the problem and measures attribute compatibility based on syntactic and semantic similarity. SANTOS~\cite{khatiwada2023santos} enriches this process with knowledge bases and column relationships, while Starmie~\cite{fan2023semantics} uses pre-trained language models to learn contextualized column embeddings for semantic matching. For joinable table search, JOSIE~\cite{zhu2019josie} retrieves top-k joinable tables via overlap set similarity search, PEXESO~\cite{dong2021efficient} embeds text values and applies similarity joins with efficient indexing, and MATE~\cite{esmailoghli2022mate} proposes XASH hash function for efficient multi-column join discovery. While unionable search focuses on semantic compatibility across attributes, and joinable search focuses on column overlap to identify joinable tables, our goal is to localize fine-grained cell-level overlaps under arbitrary row and column permutations.

\smallskip\noindent\textbf{Data Integration and Matching.}
Our problem is loosely related to schema matching and entity resolution. Schema matching aligns semantically equivalent attributes across tables, as in COMA~\cite{hong2002coma}, COMA++~\cite{aumueller2005schema}, and the recent LLM-based Magneto~\cite{liu2025magneto}. Entity resolution identifies tuples referring to the same real-world entity, such as ZeroER~\cite{wu2020zeroer} and Sudowoodo~\cite{wang2023sudowoodo}.
Unlike these methods, our work operates directly on raw tables without schema metadata or domain knowledge, and recovers exact overlapping cells.

\smallskip\noindent\textbf{Maximum Common Subgraph.}
Since no prior work addresses the maximum common subhypergraph problem, we review methods for the related maximum common subgraph problem. Non-learning approaches include McSplit~\cite{mccreesh2017partitioning}, an exact branch-and-bound algorithm that partitions the search space to prune infeasible candidate node matches, and RRSplit~\cite{yu2025fast}, which introduces tighter bounds and reductions to reduce redundant computations. Learning-based approaches extend McSplit with various techniques for node selection and branching. These include reinforcement learning method~\cite{liu2020learning}, GNN-based Deep Q-Network in GLSearch~\cite{bai2021glsearch}, and Long-Short Memory with Leaf vertex Union Match in McSplit+LL~\cite{zhou2022strengthened}. McSplitDAL~\cite{liu2023hybrid} further improves efficiency with a new value function and hybrid selection strategy. 

\smallskip\noindent\textbf{Hypergraph Modeling for Tables.}
Additionally, we highlight HyTrel~\cite{chen2023hytrel}, which uses hypergraphs to capture permutation invariance and structural properties for table representation learning. While sharing the use of hypergraphs, our construction is purpose-built for the table overlap problem, with structural constraints that enable a constrained maximum common subhypergraph search that faithfully models shape-agnostic overlaps.
\section{Conclusion and Future Work}

We introduced a novel shape-agnostic definition of largest table overlap, cast it as a maximum common subhypergraph problem via a hypergraph-based representation, and proved their equivalence. We established that the problem is NP-hard to approximate, and proposed \hyper{}, an efficient exact solver that exploits hypergraph structure for effective pruning and search. Extensive experiments demonstrate \hyper{} achieves accurate and scalable shape-agnostic overlap detection, with case studies highlighting its robustness across diverse real-world scenarios. Future work includes developing scalable approximation strategies, integrating learning-based guidance for search, and exploring richer overlap notions that capture contextual or semantic relevance beyond exact value matching.
% \input{section_tex/9_acknowledgments}

% \clearpage
\bibliographystyle{ACM-Reference-Format}
\bibliography{references}

% \clearpage
\appendix
\section*{Appendix}

\section{Prompt Template for the \gpt{} Baseline}\label{app:prompt}

We employ a large language model to predict the row and column permutations that maximize shape-agnostic table overlap. The prompt provides the tables in a structured row and column format and instructs the model to output a JSON object specifying the maximum overlap size along with the corresponding row and column permutations for each table. Figure~\ref{fig:prompt} illustrates the input format, expected output, and an illustrative task example.

% \begin{gptpromptbox}
% \begin{lstlisting}[basicstyle=\footnotesize\ttfamily, breaklines=true, breakindent=0pt, columns=fullflexible]
% \begin{figure*}[t]
% \centering
\onecolumn
\begin{evbox}
You are a combinatorial-optimization expert.

# Goal:
Given two rectangular tables of cell values (strings or numbers) you must find row and column permutations for each table that maximise the number of cells whose values match exactly in the same (row, column) position after the permutations.                                                        

Return:
a. The maximum overlap size (i.e., number of overlapping cells).
b. The four permutations (row and column for each table) that achieve this.  

Think step-by-step internally, but output only the final JSON object described in "Output format". Do not reveal intermediate reasoning.

# Input Format:
Each table is given as rows of values, with each row represented using pipes "|" to separate columns. 
                                                        
## Input Example:

Table 1: 
| A | B |
| C | D |
| E | F |

Table 2: 
| F | E |
| B | A |
| D | C |

# Output Format:
Please output only the raw JSON object without any explanation or formatting - do not wrap the output in triple backticks or add a language label. Format:
{{
  "max_overlap_size": int, // total number of matching cells after optimal permutation
  "row_perm_table1": list of int, // new row order for table1 (0-based indices)
  "col_perm_table1": list of int, // new column order for table1
  "row_perm_table2": list of int, // new row order for table2
  "col_perm_table2": list of int, // new column order for table2
}}

## Output Example:
{{
  "max_overlap_size": 5,
  "row_perm_table1": [2, 0, 1],
  "col_perm_table1": [0, 1],
  "row_perm_table2": [0, 1, 2],
  "col_perm_table2": [1, 0],
}}

# Step-by-step Example:

Starting with the example tables:
                                                        
Table 1:
| A | B |
| C | D |
| E | F |

Table 2:
| F | E |
| X | A |
| D | C |

## Step 1: Apply row permutation to Table 1: [2, 0, 1]

Table 1 rows reordered:
| E | F |  (row 2)
| A | B |  (row 0)
| C | D |  (row 1)

## Step 2: Apply column permutation to Table 1: [0, 1] (no change)

Table 1 columns stay the same:
| E | F |
| A | B |
| C | D |

## Step 3: Apply row permutation to Table 2: [0, 1, 2] (no change)

Table 2 rows stay the same:
| F | E |
| X | A |
| D | C |

## Step 4: Apply column permutation to Table 2: [1, 0]

Table 2 columns reordered:
| E | F |
| A | X |
| C | D |

## Step 5: Compare the tables cell-by-cell:

Both tables now look like:
                                               
Table 1:
| E | F |
| A | B |
| C | D |

Table 2:
| E | F |
| A | X |
| C | D |

All 5 cells (E, F, A, C, D) match exactly at same position, which is the maximum overlap size.

## Step 6: Output the results:
{{
  "max_overlap_size": 5,
  "row_perm_table1": [2, 0, 1],
  "col_perm_table1": [0, 1],
  "row_perm_table2": [0, 1, 2],
  "col_perm_table2": [1, 0],
}}

# Input:
Table 1:
{table1}

Table 2: 
{table2}

# Output:
\end{evbox}
\captionof{figure}{Prompt template for finding \salto{}}\label{fig:prompt}
\twocolumn{}
% \end{figure*}
% \end{lstlisting}
% \end{gptpromptbox}

\end{document}